\documentclass[a4paper,11pt]{article}
\usepackage[table]{xcolor}
\usepackage[utf8]{inputenc}
\usepackage{amsmath,amssymb,amsthm,graphicx,pgfplots,enumitem,comment,csquotes,subcaption,doi,multirow}
\pgfplotsset{compat=newest}
\usepgfplotslibrary{dateplot}
\usetikzlibrary{cd}
\usepackage{colortbl,xcolor,booktabs}
\usepackage{algorithm}
\usepackage{tikz,latexsym}
\usepackage[noend]{algpseudocode}
\usepackage{authblk,booktabs,tabularx}
\usepackage[margin=1in]{geometry}

\allowdisplaybreaks
\let\oldbibliography\thebibliography
\renewcommand{\thebibliography}[1]{\oldbibliography{#1}
\setlength{\itemsep}{1pt}}

\usepackage[utf8]{inputenc} 
\usepackage[T1]{fontenc}    
\usepackage{hyperref}       
\usepackage{url}            
\usepackage{booktabs}       
\usepackage{amsfonts}       
\usepackage{nicefrac}       
\usepackage{microtype}      
\usepackage{xcolor}         
\usepackage{comment}

\usepackage{hyperref}
\usepackage{amsmath,amsthm,amssymb,multirow,csquotes}
\usepackage[nameinlink]{cleveref}
\usepackage{xhfill}
\newtheorem{theorem}{Theorem}[section]

\newtheorem{proposition}[theorem]{Proposition}
\newtheorem{lemma}[theorem]{Lemma}

\newcommand{\A}{\mathbf{A}}
\newcommand{\B}{\mathbf{B}}
\newcommand{\x}{\mathbf{x}}
\newcommand{\y}{\mathbf{y}}
\newcommand{\p}{\mathbf{p}}
\newcommand{\q}{\mathbf{q}}
\newcommand{\s}{\mathbf{s}}
\newcommand{\w}{\mathbf{w}}
\newcommand{\<}{\left(}
\renewcommand{\>}{\right)}
\newcommand{\KL}{D_\text{KL}}
\newcommand{\wKL}{D_{\text{KL}(w)}}
\renewcommand{\|}{\parallel}
\newcommand{\rk}{r_k\<\x_{-k}\>}
\newcommand{\rp}{r_k\<\p_{-k}\>}

\usepackage{pgfplots,algorithm}
\usepackage[noend]{algpseudocode}

\makeatletter
\def\BState{\State\hskip-\ALG@thistlm}

\makeatother
\let\oldReturn\Return
\renewcommand{\Return}{\State\oldReturn}
\algnewcommand{\Initialize}[1]{%
  \State \textbf{Initialize:} \hspace*{0.2em}\parbox[t]{.8\linewidth}{\raggedright #1}
}

\newtheorem{property}{Property}
\theoremstyle{definition}
\newtheorem{definition}{Definition}

\theoremstyle{remark}
\newtheorem{remark}{Remark}
\crefname{property}{property}{properties}
\Crefname{property}{Property}{Properties}

\graphicspath{{./}{Matlab.files_qlearning/}{Figures/}}

\usepackage[export]{adjustbox}
\usepackage{graphicx,environ}
\usepackage{tikz}
\usepackage{tikzscale}
\usepackage{pgfplots,enumitem,color,colortbl,mathtools}
\usepackage{setspace}
\pgfplotsset{compat=newest}
\usetikzlibrary{automata,arrows,arrows.meta,calc,positioning}
\makeatletter
\newsavebox{\measure@tikzpicture}
\NewEnviron{scaletikzpicturetowidth}[1]{%
  \def\tikz@width{#1}%
  \def\tikzscale{1}\begin{lrbox}{\measure@tikzpicture}%
  \BODY
  \end{lrbox}%
  \pgfmathparse{#1/\wd\measure@tikzpicture}%
  \edef\tikzscale{\pgfmathresult}%
  \BODY
 }
\makeatother

\newcommand{\xk}{x_{ki}}

\newcommand{\N}{\mathcal N}

\newcommand{\lt}{\left[}
\newcommand{\rt}{\right]}

\tikzset{node distance=4.5cm, every state/.style={semithick,fill=gray!10}, every edge/.style={draw,->,>=stealth',auto,semithick}}

\title{Exploration-Exploitation in Multi-Agent Competition: Convergence with Bounded Rationality}

\author[1]{Stefanos Leonardos}
\author[1]{Georgios Piliouras}
\author[2]{Kelly Spendlove}
\affil[1]{Singapore University of Technology and Design, \emph{\{stefanos\_leonardos, georgios\}@sutd.edu.sg}}
\affil[2]{Mathematical Institute, University of Oxford, \emph{spendlove@maths.ox.ac.uk}}
\date{}

\begin{document}

\maketitle

\begin{abstract}
The interplay between exploration and exploitation in competitive multi-agent learning is still far from being well understood. Motivated by this, we study smooth Q-learning, a prototypical learning model that explicitly captures the balance between game rewards and exploration costs. We show that Q-learning always converges to the unique quantal-response equilibrium (QRE), the standard solution concept for games under bounded rationality, in weighted zero-sum polymatrix games with heterogeneous learning agents using positive exploration rates. Complementing recent results about convergence in weighted potential games~\cite{Cou15,Leo21}, we show that fast convergence of Q-learning in competitive settings is obtained regardless of the number of agents and without any need for parameter fine-tuning. As showcased by our experiments in network zero-sum games, these theoretical results provide the necessary guarantees for an algorithmic approach to the currently open problem of equilibrium selection in competitive multi-agent settings.
\end{abstract}

\section{Introduction}\label{sec:introduction}


Zero-sum games and variants thereof are arguably amongst the most well studied settings in game theory. Indeed much attention has focused on the class of strictly competitive games~\cite{Aum17}, i.e., two player games such that when both players change their mixed strategies, then either there is no change in the expected payoffs, or one of the two expected payoffs increases and the other decreases.\footnote{In fact, as recent work has established these strictly competitive games are formally equivalent to weighted zero-sum games, i.e., affine transformations of zero-sum games~\cite{Adl09}.} According to Aumann~\cite{Aum17},
\textit{``Strictly competitive games constitute one of the few areas in game theory, and indeed in the social sciences, where a fairly sharp, unique prediction is made."} 
The unique prediction, of course, refers to the min-max solution and the resulting values guaranteed to both agents due to the classic work of von Neumann~\cite{Neumann1928}.


Unfortunately, when we move away from the safe haven of two-agent strictly competitive games, a lot of these regularities disappear. For example, in multi-agent variants of zero-sum and strictly competitive games, several critical aspects of the min-max theorem collapse~\cite{Cai16}. Critically, Nash equilibrium payoffs
need not be unique. In fact, there could be continua of equilibria with the payoff range of different agents corresponding to positive measure sets. Furthermore, Nash equilibrium strategies need not be exchangeable (i.e., mixing-matching strategies from different Nash profiles does not lead to a Nash) or max-min. Thus, network competition is not only significantly harder, but poses qualitatively different questions than two-agent competition. \par
Nevertheless, and in spite of the intense recent interest, inspired by ML applications such as Generative Adversarial Networks (GANs) and actor-critic systems, on understanding learning dynamics in zero-sum games and even network variants thereof \cite{Das18,Gid19,Abe21}, so far there has been no systematic study of how agents should deal with uncertainty of the resulting payoffs in such games. In these settings, the use of purely optimization-driven regret-minimizing algorithms,  is no longer equally attractive as in the two-player case. The multiplicity of equilibria and the lack of unique value give rise to a non-trivial problem of \emph{equilibrium selection} and learning agents face the fundamental dilemma between exploration and exploitation \cite{Cla98,Pan05,Bus08,Zha19}. These considerations drive our motivating question:

\textit{Are there exploration-exploitation dynamics that provably converge in networks of strictly competitive games? How do they behave in settings with multiple, payoff diverse Nash equilibria?}

\paragraph{Model and Results.} We study a well known smooth variant of Q-learning~\cite{Wat92,Wat89}, with softmax or Boltzmann exploration (one of the most fundamental models of exploration-exploitation in MAS), termed Boltzmann Q-learning or \emph{smooth Q-learning}~\cite{Tuy03,Sat03}. Informally (see Section~\ref{sec:qlearning} for the rigorous definition), each agent $k$ updates their choice distribution $x=(x_i)$ according to the rule $\dot x_i/x_i= (r_i-\bar r)-T_k(\ln{x_i}-\sum_j x_j \ln{x_j})$, where $r_i,\bar r$ denote agent $k$'s rewards from action $i$ and average rewards, respectively, given all other agents' actions and $T_k$ is agent $k$'s exploration rate. \par
In our main result, we show convergence of Q-learning to Quantal Response Equilibria (QRE) (the prototypical extension of NE for games with bounded rationality \cite{Mck95}) in arbitrary networks of strictly competitive games \cite{Cai16}. As long as all exploration rates are positive, we prove via a global Lyapunov argument that the Q-learning dynamic converges pointwise to a unique QRE regardless of initial conditions and regardless of the number of the Nash equilibria of the original network game (\Cref{thm:main}). Related to the above, we demonstrate how exploration by all agents leads to equilibrium selection. This theoretically and empirically long-standing open problem (\cite{Kim96,Rom15,Per20}) becomes tractable due to the theoretical guarantees of fast convergence to QRE that we provide in \Cref{thm:main} for this class of (competitive) games.
In fact, \Cref{thm:main} is in some sense tight, as there exist network competitive settings whose dynamics lead to limit cycles as long as not all of the agents are performing exploration (see discussion and experiments in \Cref{exp:network}).


\paragraph{Other Related Works.} The variant of Q-learning that we study has recently received a lot of attention due to its connection to evolutionary game theory \cite{Sat05,Kia12,Wol12}. It has also been also extensively studied in the economics and reinforcement learning literature under various names, see e.g., \cite{Alo10,San18} and \cite{Kae96,Mer16}. Recent works demonstrate that it is possible to show convergence of the Q-learning dynamics in multi-agent \emph{cooperative} settings and to select highly desirable equilibria via bifurcations \cite{Leo21}. \par
On the other hand, competitive multi-agent systems constitute one of the current frontiers in Artificial Intelligence and Machine Learning research. Many recent works investigate the complex behavior of competitive game theoretic settings \cite{Rak13,Bal18,Bai19,Bai18,Mer19},
focusing on carefully designed convergent algorithms, e.g., ``optimism''~\cite{Das18, Das18l,wei2021linear}, extra-gradient methods~\cite{2018arXiv180702629M,antonakopoulos2020adaptive,hsieh2019convergence}, regularization~\cite{Per20}, momentum adapted dynamics~\cite{Gid19n}, or symplectic integration schemes~\cite{BaileyArxivVerlet}. 
However, despite the theoretical progress and the impressive results in the empirical front \cite{Sil16,Vin19,LAn20}, the literature on equilibrium selection seems to have received little attention so far. Prior works have focused on equilibrium selection in cooperative AI or have theoretically studied competition in 2-agent zero-sum games in which equilibrium selection is irrelevant \cite{Daf20,Das20}. To our knowledge, we are the first paper to explicitly study both theoretically and experimentally multi-agent competitive settings without uniquely defined values.

\section{Game-Theoretic Preliminaries}

A \emph{polymatrix or separable network game}, $\Gamma=\<(V,E), \<S_k,w_k\>_{k\in V},\<\A_{kl}\>_{[k,l]\in E}\>$ consists of a graph $(V,E)$, where $V=\{1,2,\dots,n\}$ is the set of players (or agents), and $E$ is a set of pairs, $[k,l]$, of players $k\neq l\in V$. Each player, $k\in V$, has a finite set of actions (or strategies) $S_k$ with generic elements $s_k\in S_k$ (depending on the context, sometimes we will also write $i$ or $j\in S_k$). Players may also use mixed strategies (or choice distributions) $\x_k=(x_{ki})_{i\in S_k}\in \Delta_k$, where $\Delta_k$ is the simplex in $\mathbb R^{|S_k|}$, i.e., $\sum_{i\in S_k} x_{ki}=1$, and $x_{ki}\ge 0$, for any $\x_k\in \Delta_k$. The \emph{interior} of $\Delta_k$ is the set of all points $\x_k\in\Delta_k$ with $x_{ki}\in(0,1)$ for all $i\in S_k$. All points of $\Delta_k$ that are not in the interior, are called \emph{boundary points}. We will write $\x =(\x_1,\x_2,\dots,\x_n)$ or $\x=(\x_k,\x_{-k})$ for a mixed strategy profile $\x \in \Delta:=\prod_{k\in V}\Delta_k$, where $\x_{-k}\in \Delta_{-k}=\prod_{l\neq k \in V} \Delta_l$ is the vector of mixed strategies of all players $l\in V$ other than $k$. \par

Each edge $[k,l]\in E$ defines a two player game with payoff matrices $\A_{kl}\in \mathbb R^{|S_k|\times |S_l|}$ and $\A_{lk}\in \mathbb R^{|S_l|\times |S_k|}$. The elements $a_{kl}(s_k,s_l)$ of a matrix $\A_{kl}$ denote the payoffs of player $k$ when players $k$ and $l$ use pure actions $s_k\in S_k$ and $s_l\in S_l$, respectively. Each player $k\in V$ chooses a strategy (mixed or pure) and plays that strategy in \emph{all} games $[k,l] \in E$. Thus, the payoff of player $k$ at the pure strategy profile $\s =(s_1,\dots,s_n) \in S:= \prod_{k\in V}S_k$ is $u_k(\s)=\sum_{[k,l]\in E}a_{kl}(s_k,s_l)$. Similarly, the expected reward of player $k$ in the mixed strategy profile $\x\in \Delta$ is  
\begin{equation}\label{eq:utility}
u_k(\x):=\sum\nolimits_{[k,l]\in E} \x_k^\top \A_{kl} \x_l = \x_k^\top \<\sum\nolimits_{[k,l]\in E}\A_{kl}\x_l\>.
\end{equation}
It will be convenient to write $r_{ki}(\x_{-k}):= u_k\<i,\x_{-k}\>=\sum_{[k,l]\in E}\{\A_{kl}\x_l\}_i$ (where $\{v\}_i$ to denotes the $i$-th element of a vector $v$), for the reward of pure action $i\in S_k$ of player $k$ when all other players use the strategy profile $\x_{-k}$ and $\rk:=(r_{ki}(\x_{-k}))_{i\in S_k}$ for the resulting reward vector of all pure actions of agent $k$, respectively. Using this notation, the expected reward of player $k\in V$ at the mixed strategy profile $\x=(\x_k,\x_{-k})$ can be compactly expressed as $u_k(\x)=\x_k^\top \rk$.\par
\paragraph{Weighted Zero-Sum Polymatrix Games.} $\Gamma$ is called a \emph{weighted or rescaled} zero-sum polymatrix game \cite{Cai16}, if there exist positive constants $w_1,w_2,\dots,w_n>0$, so that
\begin{equation}\label{eq:zsprop}
\sum\nolimits_{k\in V}w_ku_k(\x)=0, \;\; \text{for all }\x \in \Delta.
\end{equation}
By summing over the edges in $E$ (instead of the players in $V$), we can equivalently express the weighted zero-sum property as 
\begin{equation}\label{eq:zsequivalent}
\sum\nolimits_{[k,l]\in E}\lt w_k\x_k^\top \A_{kl} \x_l +w_l \x_l^\top \A_{lk} \x_k\rt= 0, \;\; \text{for all }\x \in \Delta.
\end{equation}
\paragraph{Nash Equilibrium.} A strategy profile (tuple of mixed strategies), $\p=(p_k)_{k\in V}\in \Delta$, with $\p_k=(p_{ki})_{i\in S_k} \in \Delta_k$ for each $k\in V$ is a \emph{Nash equilibrium} of $\Gamma$ if 
\begin{equation}\label{eq:knash}
u_k(\p)\ge u_k(x_k,\p_{-k}), \quad \text{for all }x_k \in \Delta_k, \; \text{for all } k\in V,
\end{equation}
i.e., if there exist no profitable unilateral deviations. By linearity, it suffices to verify the condition in equation \eqref{eq:knash} only for pure actions $s_k\in S_k$ instead of all $x_k\in \Delta_k$. 

\section{Joint-learning Model: Q-learning Dynamics}\label{sec:qlearning}

We next discuss how we can get to the Q-learning dynamics from Q-learning agents when there are multiple learners in the system. Our goal is to identify the dynamics in competitive systems in which multiple agents are playing a rescaled zero-sum polymatrix game repeatedly over time. \par
\paragraph{Q-learning.} Q-learning \cite{Wat92,Wat89} is a value-iteration method for solving the optimal strategies in Markov Decision Processes (MDPs). It can be used as a model where users learn about their optimal strategy when facing uncertainties. In Q-learning, each agent $k\in \N$ keeps track of the past performance of their actions $i\in S_k$ via the Q-learning update rule
\begin{equation}\label{eq:update_q}
Q_{ki}(n+1)=Q_{ki}(n)+\alpha_k\lt r_{ki}(\x_{-k},n)-Q_{ki}(n)\rt,\;i \in S_k
\end{equation}
where the additional argument, $n\ge0$, denotes (discrete) time steps and $\alpha_k\in[0,1]$ denotes the learning rate or memory decay of agent $k$, cf. \cite{Sat03,Kia12}. $Q_{ki}(t)$ is called the \emph{memory} of agent $k$ about the performance of action $i\in S_k$ up to time step $t\ge0$. Agent $k\in V$ updates their actions (choice distributions) according to a Boltzmann type distribution, with 
\begin{equation}\label{eq:update_x}
x_{ki}(n)=\frac{\exp{(Q_{ki}(n)/T_k)}}{\sum_{j\in S_k}\exp{(Q_{kj}(n)/T_k)}},\quad \text{for each } i\in S_k,
\end{equation}
where $T_k\in[0,+\infty)$ denotes agent $k$'s learning sensitivity or adaptation, i.e., how much the choice distribution is affected by the past performance. We will refer to $T_k$ as the \emph{exploration rate} or \emph{temperature} of player $k$ \cite{Leo21,Tuy03} (see \Cref{rem:rates} for a discussion). Combining equations \eqref{eq:update_q} and \eqref{eq:update_x}, one obtains the recursive equation of player $k$'s mixed strategy (or choice distribution)
\begin{align*}
&x_{ki}(n+1)=\frac{x_{ki}(n)\exp{((Q_{ki}(n+1)-Q_{ki}(n))/T_k)}}{\sum_{j\in S_k}x_{kj}(n)\exp{((Q_{kj}(n+1)-Q_{kj}(n))/T_k)}}\,,
\end{align*}
for each $i\in S_k$. In practice, agents perform a large number of actions (updates of Q-values) for each choice-distribution (update of Boltzmann selection probabilities). This motivates to consider a continuous time version of the learning process for each agent $k\in V$ which results in the following update rules for both the memories $Q_{ki}$ and the selection probabilities $x_{ki}$ of each action $i\in S_k$
\[\dot Q_{ki}=\alpha_k\lt r_{ki}(\x_{-k})-Q_{ki} \rt, \; \text{ and }\; \dot x_{ki}=x_{ki}\left(\dot Q_{ki}-\sum\nolimits_{j\in S_k}\dot Q_{kj}\right)/T_k,\]
where we suppressed the dependence on (the continuous) time steps $t\ge0$. Combining the last two equations under the assumptions that the relationship between pairs of actions is constant over time and that the choice distributions of the various agents are independently distributed, yields the \emph{Q-learning dynamics}
\begin{equation}\label{eq:kdynamics}
\dot{x}_{ki}=\xk\lt r_{ki}(\x_{-k})-\x_k^\top\rk -T_k\< \ln{(\xk)}-\x_k^\top\ln{(\x_k)} \> \rt,
\end{equation}
for all $i \in S_k$ and all players $k\in V$. The next property provides intuition for the meaning of Q-values in a practical context (see also Remark \ref{rem:rates} for a discussion). As is the case for all technical materials of this section, the proof of \Cref{lem:qcontinuous} is deferred to Appendix A. 

\begin{proposition}[Q-value Updates and Exponential Discounting]\label{lem:qcontinuous}
Consider the Q-learning updates
\begin{equation}\label{eq:qupdates}
Q_i(n+1)=Q_i(n)+\alpha[r_i(n)-Q_i(n)]
\end{equation}
where $n\ge0$ are discrete time steps and assume that $Q_i(0)=0$ for all $i\in A$. Then, in continuous time, the updates are given by
\[Q_i(t)=\alpha \int_{0}^te^{-\alpha s}r_i(t-s) ds =\alpha e^{-\alpha t}\int_{0}^te^{\alpha s}r_i(s) \mathop{ds}, \quad \text{for any } t>0.\]
\end{proposition}

\begin{remark}\label{rem:rates}
Equation \eqref{eq:update_x} implies that higher values of $T_k$ indicate a higher exploration rate, i.e., randomization among the agent's available choices in $S_k$, whereas values of $T_k$ close to $0$ indicate a higher exploitation rate, i.e., preference of the agent towards the best performing action. Specifically, for $T_k\to0$, \eqref{eq:update_x} corresponds to the best-response update rule, and agent $k$ selects the action $i\in S_k$ with the highest $Q$-value with probability one. On the other hand, as $T_k\to \infty$, equation \eqref{eq:update_x} corresponds to the case in which player $k$ selects their actions uniformly at random. \par
From a behavioral perspective, this leads to an equivalent interpretation of the $T_k$'s as the degrees of (bounded) rationality of the agents \cite{EWA1,EWA2,Gal13}. Similarly, \Cref{lem:qcontinuous} implies that the Q-learning updates can be intuitively interpreted as exponential discounting in the payoffs, with payoffs further back in the past being less important. In other words, the Q-value of an agent for an action $i$ shows how beneficial the (discounted) strategy $i$ is for that agent. On the other hand, from a purely algorithmic perspective, the $T_k$'s and the corresponding entropy terms in the players' choice distributions updates can be viewed as regularization terms that prevent the learning dynamics from reaching the boundary or from getting trapped in local optima, see \cite{Bow02,Cou15,Mer18} and \cite{Kai10,Kai11}.
\end{remark}

\paragraph{Quantal Response Equilibria.}
Using the convention $x\ln{x}:=0$ whenever $x=0$ (recall that $\lim_{x\to 0^+}x\ln{x}=0$), it is immediate to infer from equation \eqref{eq:kdynamics} that $x_{ki}(t)=x_{ki}(0)$ for all $t>0$, for all $i\in S_k$ such that $x_{ki}(0)\in\{0,1\}$. In other words, if the dynamics start on the boundary, then they will remain there, i.e., the boundary of $\Delta_k$ is invariant for the Q-learning dynamics. This implies, in particular, that all corner points of the simplex $\Delta_k$, i.e, all points $x_k=(x_{ki})_{i\in S_k}$ with one coordinate equal to $1$ (and consequently, all other coordinates equal to $0$) are (trivially) fixed points of the Q-learning dynamics. Thus, the interesting part concerns the fixed-point analysis of the dynamics for \emph{interior} starting points. \par
The \emph{interior} fixed points of the Q-learning dynamics are called \emph{Quantal Response Equilibria}. From equation \eqref{eq:kdynamics}, it is apparent that whenever $T_k>0$ for all $k\in V$, then such points can only lie in the interior of $\Delta_k$ due to the entropy term ($\ln$ of the choice probabilities). Thus, at such fixed points the rate of change of $x_{ki}$ must be equal to zero. Specifically, a strategy profile $\p=(p_k)_{k\in V}$ with $\p_k=(p_{ki})_{i\in S_k}$ for each $k\in V$ is a \emph{Quantal Response Equilibrium (QRE)} of $\Gamma$ if
\begin{equation}\label{eq:kqre}
r_{ki}(\p_{-k})=\p_k^\top \rp +T_k\< \ln{(p_{ki})}-\p_k^\top\ln{(\p_k)} \> , \;\; \text{for all } k\in V.
\end{equation}
The following properties are immediate from the characterization of the QRE in equation \eqref{eq:kqre}.
\begin{theorem}[Interior Fixed Points of the Q-learning Dynamics]\label{thm:preliminary} Let $\Gamma$ be an arbitrary game, with positive exploration rates $T_k$ and consider the associated Q-learning dynamics
\[\dot\xk=\xk\lt r_{ki}(\x_{-k})-\x_k^\top\rk -T_k\< \ln{(\xk)}-\x_k^\top\ln{(\x_k)} \> \rt,\;\; i\in S_k, k\in V.\] 
The interior fixed points, $\p=(p_k)_{k\in V}$ of the Q-learning dynamics are the solutions of the system  
\begin{equation}\label{eq:qre}
p_{ki}=\frac{\exp{\<r_{ki}(\p_{-k})/T_k\>}}{\sum_{j\in S_k}\exp{\<r_{kj}(\p_{-k})/T_k\>}}, \quad \text{for all } i\in S_k.
\end{equation}
Such fixed points always exists and coincide with the Quantal Response Equilibria (QRE) of $\Gamma$. Given any such fixed point $\p$, we have, for all $\x_k\in \Delta_k$ and for all $k\in V$, that
\begin{equation}\label{eq:eqprop}
(\x_k-\p_k)^\top \lt r_{k}(\p_{-k})-T_k \ln{(\p_k)}\rt=0.
\end{equation}
\end{theorem}

The proof of \Cref{thm:preliminary} does not make use of the weighted zero-sum property, cf. equation \eqref{eq:zsprop}, and thus, it holds for arbitrary network games (not necessarily weighted network zero-sum). 

\paragraph{Exploration Rates.} 
Equation \eqref{eq:eqprop} in \Cref{thm:preliminary} is critical for the main technical step in our proof of convergence of the Q-learning dynamics in weighted zero-sum polymatrix games. Thus, it is important to note that this equation only holds when \emph{all} exploration rates, $T_k,k\in V$ are positive (and not merely non-negative). The intuition is that when all $T_k>0$, every QRE becomes interior, i.e., has full support. This implies that all deviations (to pure or mixed strategies) give the \emph{same} expected reward to the deviating agent. If some $T_k$ are equal to $0$, then \eqref{eq:eqprop} holds with (non-strict) inequality which is not generally enough for the proof of \Cref{thm:main}.\par

When $T_k=0$ for all $k\in V$, the Q-learning dynamics reduce to the standard replicator dynamics and the QRE coincide with the Nash equilibria of $\Gamma$, cf. equation \eqref{eq:knash}. The case in which only some $T_k$ are equal to $0$ exhibits more interesting behavior and will be discussed in detail later. In brief, if this is the case, then the dynamics may also converge to a boundary point (for the non-exploring players, i.e., for the players $k\in V$ with $T_k=0$), even if the dynamics start in the interior.

\section{Convergence of Q-learning in Weighted Zero-sum Polymatrix Games}\label{sec:convergence}

Next, we consider the joint learning model (Q-learning) of the previous section in rescaled zero-sum polymatrix games, $\Gamma$, as defined above. Our main result in this section is that Q-learning converges to the QRE of $\Gamma$. The key step of the proof is to show that the \emph{distance} between an interior QRE, $\p\in \Delta$ and the sequence of play, $\x(t),t\ge0\in \Delta$, that is generated by the Q-learning model is monotonically decreasing. To measure this distance in a meaningful way, we will use the notion of KL-divergence which is formally defined next.

\begin{definition}[Kullback-Leibler (KL) Divergence]\label{def:kldivergence} 
The Kullback-Leibler or \emph{KL-Divergence} (also called \emph{relative entropy}), $\KL$, between two strategy profiles $\p=(\p_k)_{k\in V}, \x(t)=(\x_k(t))_{k\in V} \in \Delta$ with $\p_k=(p_{ki})_{i\in S_k}$ and $\x_k(t)=(\xk(t))_{i\in S_k} \in \Delta_k$ for all $k\in V$, is defined as 
\begin{equation}\label{eq:defkl}
\KL\<\p \|\x(t)\>:=\sum\nolimits_{k\in V}\KL(\p_k\|\x_k(t)) =\sum\nolimits_{k\in V}\p_{k}^\top \ln{\<\frac{\p_{k}}{\x_{k}(t)}\>},
\end{equation}
where $\p_{k}^\top \ln{\<\frac{\p_{k}}{\x_{k}(t)}\>}=\sum_{i\in S_k}p_{ki} \ln{(p_{ki}/x_{ki}(t))}$. If $\w=(w_k)_{k\in V}$ is a $k$-dimensional vector of positive scalars, then the \emph{weighted or rescaled KL-divergence}, $\wKL$, is defined as 
\begin{equation}\label{eq:weightedkl}
\wKL\<\p \|\x(t)\>:=\sum\nolimits_{k\in V}w_k\KL(\p_k\|\x_k(t)).
\end{equation}
\end{definition}
The KL-divergence between $\p$ and $\x(t)$ can be thought of as a measurement of how far the distribution $\p$ is from the distribution $\x(t)$. However, the KL-divergence is not symmetric, i.e., in general it holds that $\KL(\p\|\x(t))\neq\KL(\x(t)\|\p)$. 



\paragraph{Main Result.} Using the notion of KL-divergence we can now formulate our main result.

\begin{theorem}[Convergence of Q-learning dynamics]\label{thm:main}
Let $\Gamma$ be a rescaled zero-sum polymatrix game, with positive exploration rates $T_k$.  There exists a unique QRE $\p$ such that if $\x(t)$ is any trajectory of the associated Q-learning dynamics $\dot \x = f(\x)$,  with $f_i$ is given via \eqref{eq:kdynamics}, where $\x(0)$ is an interior point, then $\x(t)$ converges to $\p$ exponentially fast.  In particular, we have that 
\begin{equation}\label{eqn:KLlyapunov}
\frac{d}{dt}\wKL(\p\|\x(t))=-\sum\nolimits_{k\in V}w_kT_k\lt \KL(\p_k\|\x_k)+\KL(\x_k\|\p_k) \rt.
\end{equation}
\end{theorem}
\begin{proof}[Sketch of Proof]
Let $\p=(\p_k)$ be a QRE equilibrium for the Q-learning dynamics (existence of $\p$ is guaranteed via \Cref{thm:preliminary}), and let $\x(t)$ be a trajectory with $\x(0)$ an interior point. We will first establish  \eqref{eqn:KLlyapunov}, from which the other statements follow.  The first step is to derive an explicit formula for the time derivative of the KL divergence between $\p_k$ and $\x_k(t)$ for a given player $k$.
\begin{lemma}\label{lem:tderivative}
For any QRE equilibrium $\p$ and any player $k\in V$ we have that 
\begin{equation*}
\frac{d}{dt}\KL(\p_k\|\x_k(t))=(\x_k-\p_k)^\top \lt \rk-r_{k}(\p_{-k})\rt-T_k\lt \KL(\p_k\|\x_k)+\KL(\x_k\|\p_k) \rt .
\end{equation*}
\end{lemma}
 It follows from Lemma~\ref{lem:tderivative} that
\begin{align}
&\frac{d}{dt}\KL(\p\|\x(t))=~\sum_{k\in V}\frac{d}{dt}w_k\KL(\p_k\|\x_k) \nonumber\\
&=~\sum_{k\in V}w_k\lt (\x_k-\p_k)^\top \lt \rk-r_{k}(\p_{-k})\rt-T_k\< \KL(\p_k\|\x_k)+\KL(\x_k\|\p_k) \>\rt \nonumber\\
&=~\sum_{k\in V}w_k(\x_k-\p_k)^\top \lt \rk-r_{k}(\p_{-k})\rt-\sum_{k\in V}w_kT_k\lt \KL(\p_k\|\x_k)+\KL(\x_k\|\p_k) \rt.\label{eqn:proof:DKL}
\end{align}
Notice that \eqref{eqn:KLlyapunov} follows from \eqref{eqn:proof:DKL} if the first term vanishes.  This can be shown by utilizing the following result regarding rescaled zero-sum polymatrix games.
\begin{lemma}\label{lem:summation}
Given any QRE equilibrium $\p=(\p_k)$ and $\x=(\x_k)$ we have that
\begin{equation}\label{eq:summation}
\sum\nolimits_{k\in V}w_k\lt \x_k^\top \rp+\p_k^\top\rk\rt=0.
\end{equation}
\end{lemma}
Returning to \eqref{eqn:proof:DKL}, we have that the first term in the right hand side of the last equation vanishes since
\begin{align*}
\sum_{k\in V}w_k(\x_k-\p_k)^\top \lt \rk-r_{k}(\p_{-k})\rt&~=~\sum_{k\in V}w_k\x_k^\top \rk+\sum_{k\in V}w_k\p_k^\top r_k(\p_{-k})\\
&~~+\sum_{k\in V}w_k\lt \x_k^\top \rp+\p_k^\top\rk\rt=0,
\end{align*}
where the last equality follows from both the zero-sum property \eqref{eq:zsprop} and \Cref{lem:summation}. Thus,
\begin{align*}
\frac{d}{dt}\KL(\p\|\x(t))=-\sum\nolimits_{k\in V}w_kT_k\lt \KL(\p_k\|\x_k)+\KL(\x_k\|\p_k) \rt.
\end{align*}
This equality implies that $\wKL$ is a Lyapunov function for the Q-learning dynamics.  Moreover, it implies that $\p$ is unique, and that $\x(t)$ is converging to $\p$ at an exponential rate.
\end{proof}

\paragraph{Equilibrium Selection.}
Theorem~\ref{thm:main} suggests a tractable, algorithmic approach to the problem of equilibrium selection in weighted zero-sum polymatrix games.  In particular, given an interior initial condition and positive exploration rates $T_k$, Theorem~\ref{thm:main} guarantees the system will converge to a unique QRE.  Moreover, it will do so exponentially fast.  Thus, one can use the following exploration policy set up a fast-slow system for equilibrium selection:  for each player, the exploration rates $T_k$ are brought to zero slowly (the slow system), while the Q-learning dynamics (the fast system) converge at an exponentially rate to the unique QRE; this process will approximate a Nash equilibrium of $\Gamma$.

\section{Experiments: Equilibrium Selection in Competitive Games}\label{exp:zerosum}

We study the performance of Q-learning in (weighted) network zero-sum and (as a subcase thereof) in 2-agent zero-sum games. We also investigate numerically the case in which only \emph{some} of the agents have positive exploration rates.

\paragraph{Two-agent Weighted Zero-Sum Games.}

The visualizations that we obtain in this low-dimensional case allow us to build intuition that carries over to the higher dimensional cases (with more agents) that we treat later.\\[0.2cm]
\textbf{Experimental setup:} We consider the Asymmetric Matching Pennies (AMPs) (a variation of the well-known Matching Pennies game \cite{Wun10}) which is a 2-agent, weighted zero-sum game. Each agent has two actions $\{H,T\}$ and the payoff matrices are $
\A=\begin{pmatrix}2 & -2\\ 0 & \phantom{-}2\end{pmatrix}, \B=\begin{pmatrix} \phantom{-}4 & \phantom{-}0 \\ -4 & -4\end{pmatrix}$. The AMPs game is a weighted zero-sum game since $\A+0.5\cdot \B^{\top}=0$ (each agent is the row agent in their matrix), with a unique interior Nash equilibrium, $(\p,\q)=((1/3,2/3), (2/3,1/3))$.

\paragraph{Results:} In \Cref{fig:surface}, we visualize the QRE surfaces (light blue manifolds) for different exploration rates $T_x,T_y$ of the two agents ($x-y$ plane). 
\begin{figure}[!tb]
    \centering
    \includegraphics[width=0.42\linewidth]{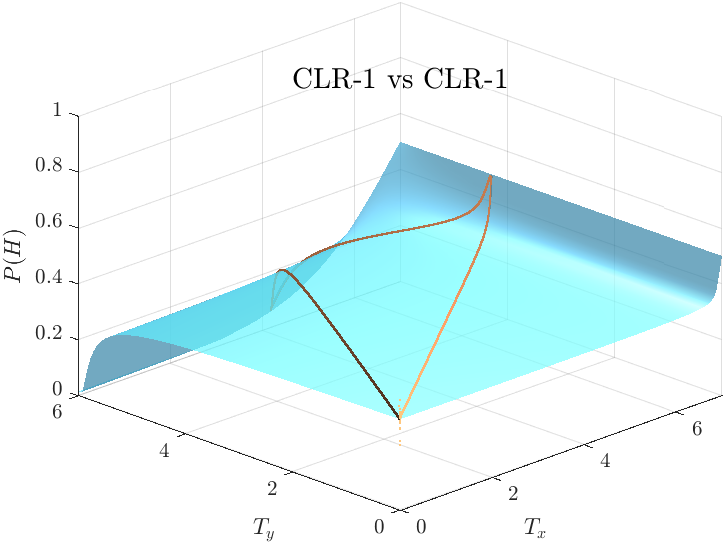}\hfill
    \raisebox{-0.5cm}{\includegraphics[width=0.49\linewidth]{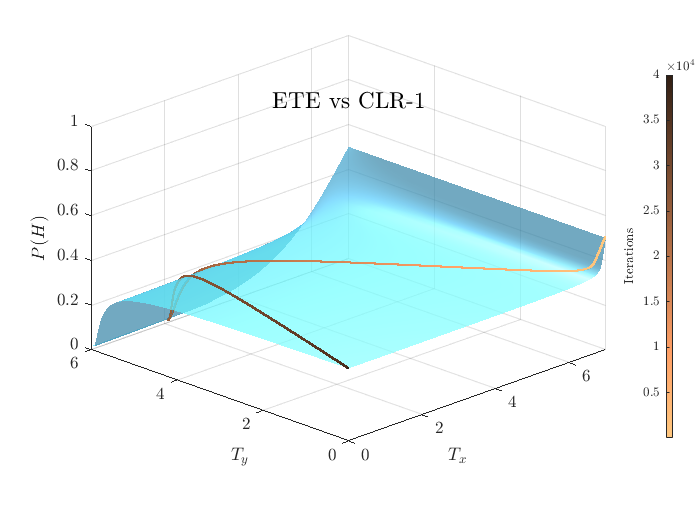}}\\
    \includegraphics[width=0.23\linewidth]{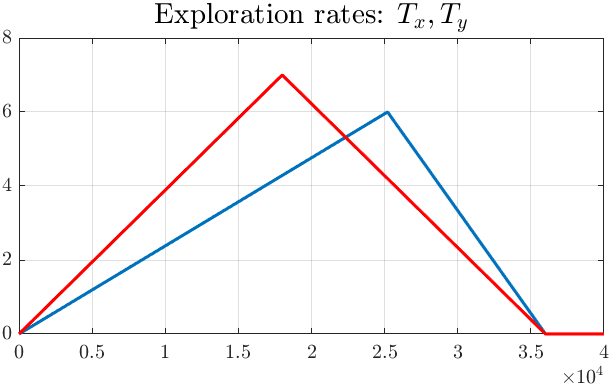}\hspace{2pt}
    \includegraphics[width=0.23\linewidth]{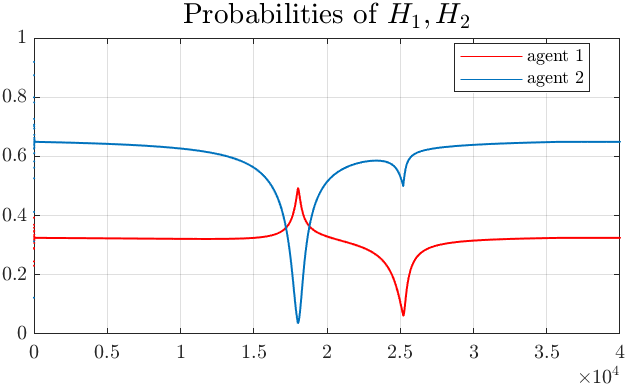}\hfill
    \includegraphics[width=0.23\linewidth]{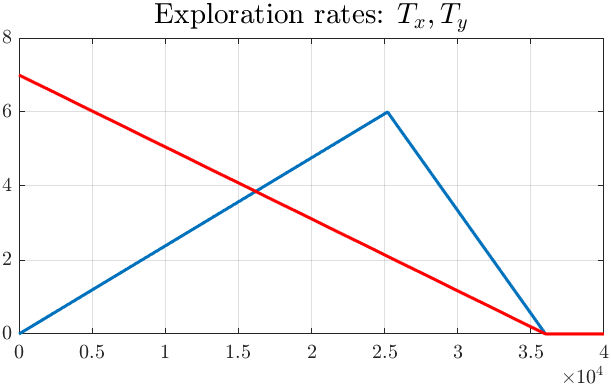}\hspace{2pt}
    \includegraphics[width=0.23\linewidth]{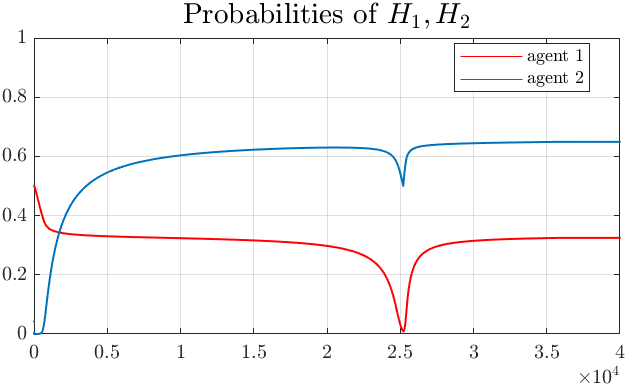}
    \caption{QRE surface and exploration paths (upper panels) for two different exploration policy profiles (lower panels) in the Asymmetric MPs game. For any combination of exploration policies (CLR-1 and ETE), the sequence of play converges to the unique QRE and as the exploration rates decrease zero, the sequence of play converges to the unique Nash equilibrium of the game.}
    \label{fig:surface}
\end{figure}
The vertical axis shows the probability with which agent $1$ chooses $H$ at the unique QRE of the game. We plot the exploration path along two representative exploration-exploitation policies: \emph{Explore-Then-Exploit} (ETE) \cite{Bai20}, which starts with (relatively) high exploration that gradually reduces to zero and \emph{Cyclical Learning Rate with one cycle} (CLR-1) \cite{Smi17}, which starts with low exploration, increases to high exploration around the half-life of the cycle and then decays to zero. For each pair of exploration rates, the learning dynamics converge to the unique QRE that corresponds to these exploration rates. As the exploration rates decay to zero, the dynamics converge to (i.e., select) the unique QRE (i.e, the Nash equilibrium in this case) of the original game. 
\begin{figure}[b]
    \centering
    \includegraphics[width=0.24\linewidth]{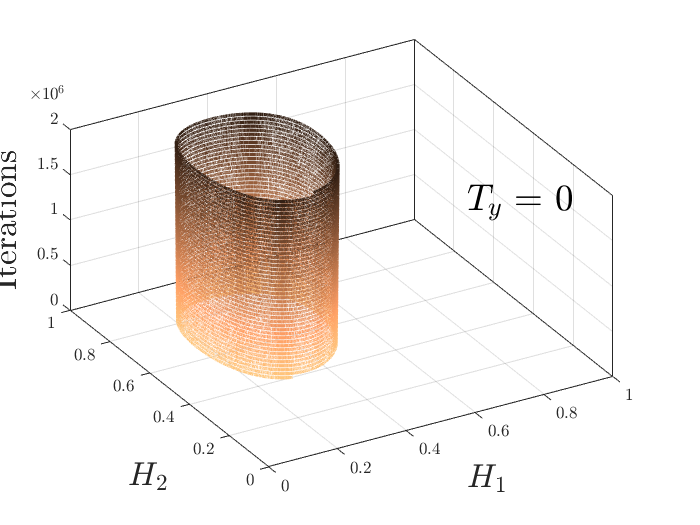}\hspace{-3pt}
    \includegraphics[width=0.24\linewidth]{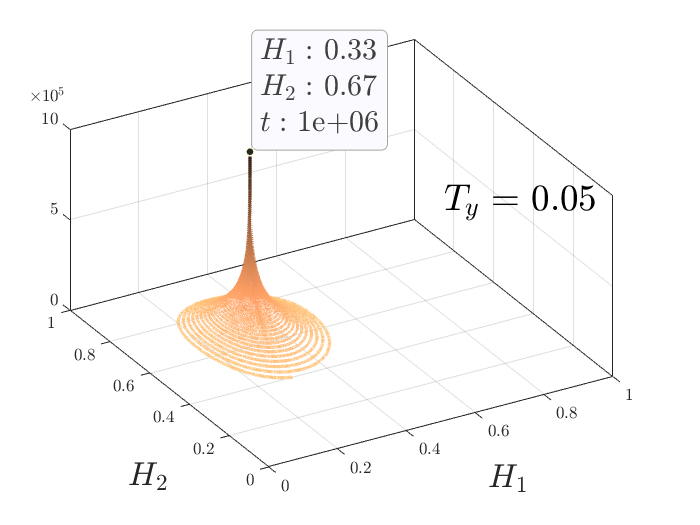}\hspace{-3pt}
    \includegraphics[width=0.24\linewidth]{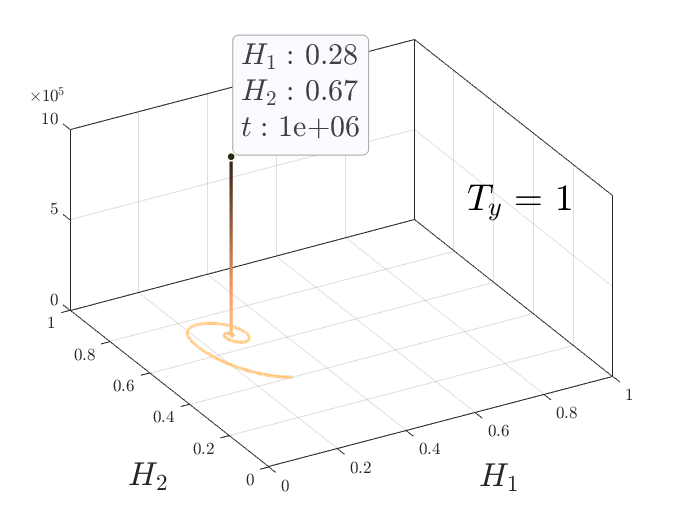}\hspace{-3pt}
    \includegraphics[width=0.24\linewidth]{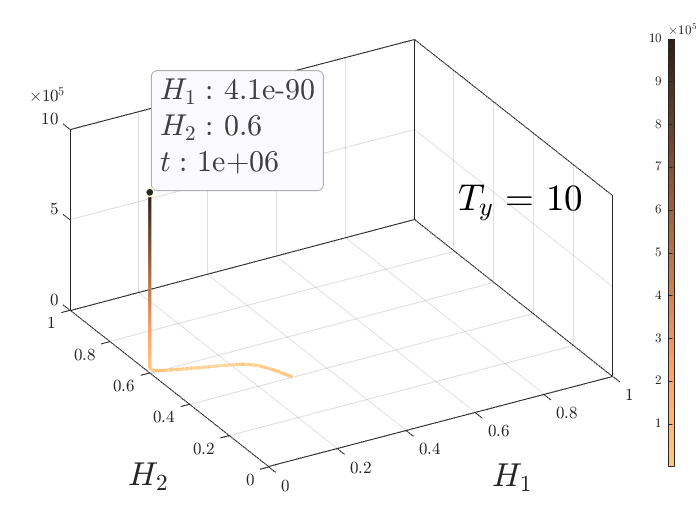}
    \caption{Q-learning dynamics in the AMPs game for $T_x=0$ (no exploration by $x$-agent) and four different exploration rates, $T_y\ge0$ by $y$-agent.}
    \label{fig:zero}
\end{figure}

By contrast, the first panel of \Cref{fig:zero} shows that the dynamics cycle around the unique Nash equilibrium when both agents do not use exploration. The rest of the panels of \Cref{fig:zero} show that in this case, exploration by only one agent suffices to lead the joint-learning dynamics to convergence. In this case, for higher values of exploration by the exploring agent, the QRE component of the non-exploring agent may lie at the boundary (see panel 4 in \Cref{fig:zero} and Appendix C).


\paragraph{Network Zero-Sum Games.}\label{exp:network}

As the next experiment shows, equilibrium selections works also in larger networks, provided that all agents maintain a positive exploration rate. In contrast to the previous example, in this case, exploration by only some agents is not sufficient for convergence.

\paragraph{Experimental setup:} We consider a zero-sum polymatrix game (ZSPG) with $n+2$ agents ($n\in \mathbb{N}$ is arbitrary) which is depicted in \Cref{fig:network_zero_sum}. 
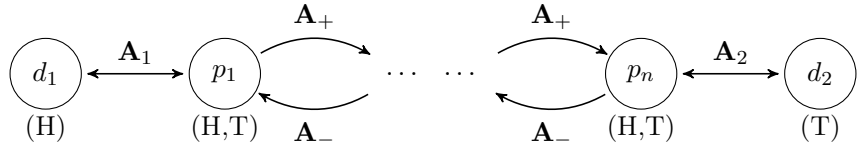
\begin{figure}[!htb]\centering
\begin{scaletikzpicturetowidth}{0.7\linewidth}
\begin{tikzpicture}[state/.style={circle, draw, minimum size=1cm}, scale=\tikzscale, 
every node/.append style={transform shape}]
\node[state, label={[shift={(0,-1.6)}] (H) }] (d1) {$d_1$};
\node[state, right =1.5cm of d1, label={[shift={(0,-1.6)}] (H,T) }] (p1) {$p_1$};
\node[state, draw=none, right =1.5cm of p1] (dt1) {$\dots$};
\node[state, draw=none, right =-0.2cm of dt1] (dt2) {$\dots$};
\node[state, right =1.5cm of dt2,label={[shift={(0,-1.6)}] (H,T) }] (pn) {$p_n$};
\node[state, right =1.5cm of pn, label={[shift={(0,-1.6)}] (T) }] (d2) {$d_2$};

\draw[shorten <=2pt, shorten >=2pt]
(d1) edge[<->] node[above, midway] (e1) {$\A_1$} (p1)
(pn) edge[<->] node[above, midway] (e7) {$\A_2$}(d2);

\draw[every loop,bend left,shorten <=2pt] 
(p1) edge node[above, midway] (e3) {$\A_+$}(dt1)
(dt2) edge node[above, midway] (e5) {$\A_+$}(pn)
(dt1) edge node[below=0.05cm, midway] (e4){$\A_-$} (p1)
(pn) edge node[below=0.05cm, midway] (e6) {$\A_-$} (dt2);
\end{tikzpicture}
\end{scaletikzpicturetowidth}
\caption{The match-mismatch network zero-sum game.}\label{fig:network_zero_sum}
\end{figure}
Each agent $p_1,p_2,\dots,p_n$ has two actions $\{H,T\}$ and receives $+1$ if they match the action of the next agent (otherwise they receive $-1$) and $+1$ if they mismatch the strategy of the previous agent (otherwise they receive $-1$). There are two dummy agents, $d_1,d_2$, who have the same payoffs but only one action: H for $d_1$ and $T$ for $d_2$. In particular, the payoff matrices for the games between non-dummy agents are given by 
\begin{equation}\label{eq:network}
\A_+=\begin{pmatrix} \phantom{-}1 & -1\\
 -1 & \phantom{-}1\end{pmatrix},\quad \A_-=-\A_+ \tag{ZSPG}
\end{equation}
whereas the payoff matrices of dummy agents $d_1$ and $d_2$ are given by $\A_1=\A_2= (1, -1)$. The column agents in the payoff matrices $\A_1$ and $\A_2$ are $p_1$ and $p_n$, and their payoffs in the encounters against $d_1$ and $d_2$ are given by $-\A_1^\top$ and $-\A_2^\top$, respectively. This game has multiple (infinite many) Nash equilibria of the following form: the odd-numbered agents play pure strategy $T$ (which, for agent 1 ensures $+1$ against $d_1$) whereas the even agents are indifferent between $H,T$ (since they are certain to have the opposing result in the two games against the previous and the next agent, both being odd and playing $T$).

\paragraph{Results:}
We provide visualizations from two instances of the \eqref{eq:network} game. In the first, with $n=3$ non-dummy agents (cf. panels 1-3 in \Cref{fig:line_network}), we visualize the projections of the QRE surfaces on the $H$ coordinate at QRE for agents $1$ to $3$ for fixed $T_3=3$ and all possible combinations of $T_1,T_2$ (x-y planes). The QRE manifolds approach continuously the boundary (case with $T_i=0$ for $i=1,2$) which leads to a unique equilibrium selection (also when $T_1=T_2=0$). The fourth panel shows a snapshot of the Lyapunov function (KL-divergence between a choice distribution $\x$ and the unique QRE $\q$) for a fixed exploration profile $T_k>0$ for all $k=1,\dots, n$ in an instance with $n=7$ non-dummy agents (the plot is similarly for all $n$). For this plot, we use the dimension reduction technique (to visualize high-dimensional surfaces along two randomly chosen directions) of \cite{Li18} (see also Appendix C). As expected, the KL-divergence is convex and decreasing for all $\x$ with a unique minimizer.

\begin{figure}[b]
    \centering
    \includegraphics[width=0.95\linewidth,clip=true, trim=5.2cm 0 4.3cm 0]{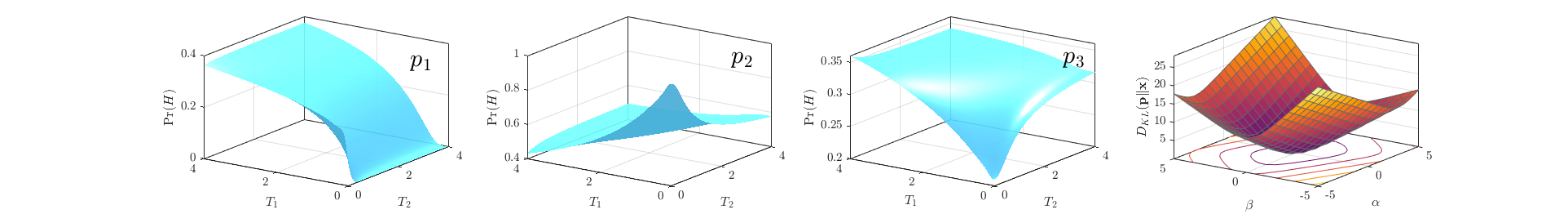}
    \caption{Projections of the QRE surfaces on the $H$ coordinate of players $p_1,p_2,p_3$ (panels 1-3) and a snapshot of the Lyapunov function, KL-divergence, (panel 4) in two instances of the \eqref{eq:network} game with $3$ and $7$ non-dummy agents, respectively, as discussed in \Cref{exp:network}.}
    \label{fig:line_network}
\end{figure}

\Cref{fig:summary} shows summary statistics for the $n=7$ instance of the \eqref{eq:network} game for three different exploration profiles. The panels in the leftmost column show the equilibria (top) and utilities (bottom) of the $7$ agents in $100$ runs when $T_k=0$ for all $k=1,2,\dots,n$ (no exploration). In this case, the dynamics converge in all runs to the pure action for the odd agents and to some arbitrary (and different every time) mixed action for the even agents. This behavior of the learning dynamics is in line with previous results in adversarial learning when equilibria lie on one face (relative boundary) of the high-dimensional simplex (cf. \cite{Mer18}).\par
The panels in the middle column correspond to a case with exploration by agents $1,2,3$ and no exploration by agents $4,\dots,7$. Exploration by the first group of agents ensures convergence of their individual dynamics to the corresponding component of the QRE (in line with the theoretical prediction of \Cref{thm:main}) but even-numbered agents that lie further away from that group fail to converge to a unique outcome (boxplot for agents 4 and 6). The outcome in this case, shows that the statement of \Cref{thm:main} is tight, and exploration by only a subset of agents may not be enough to ensure a unique outcome in general settings. Finally, the panels in the rightmost column show convergence to a unique QRE when all agents have positive exploration rates.

\begin{figure}[!tb]
    \centering
    \includegraphics[width=0.327\linewidth,valign=b]{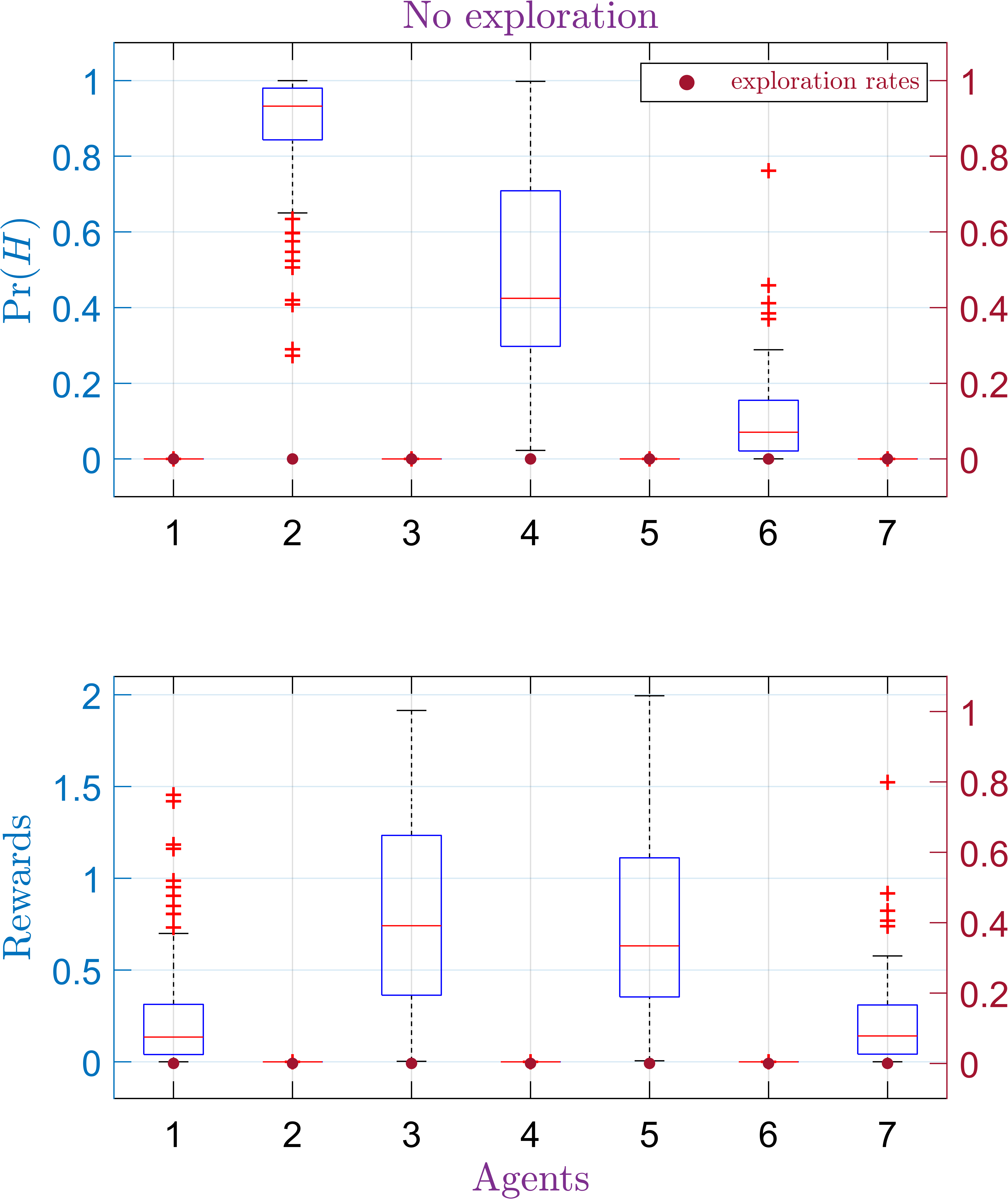}\hspace{0.2cm}
    \raisebox{0em}{\includegraphics[width=0.309\linewidth,valign=b]{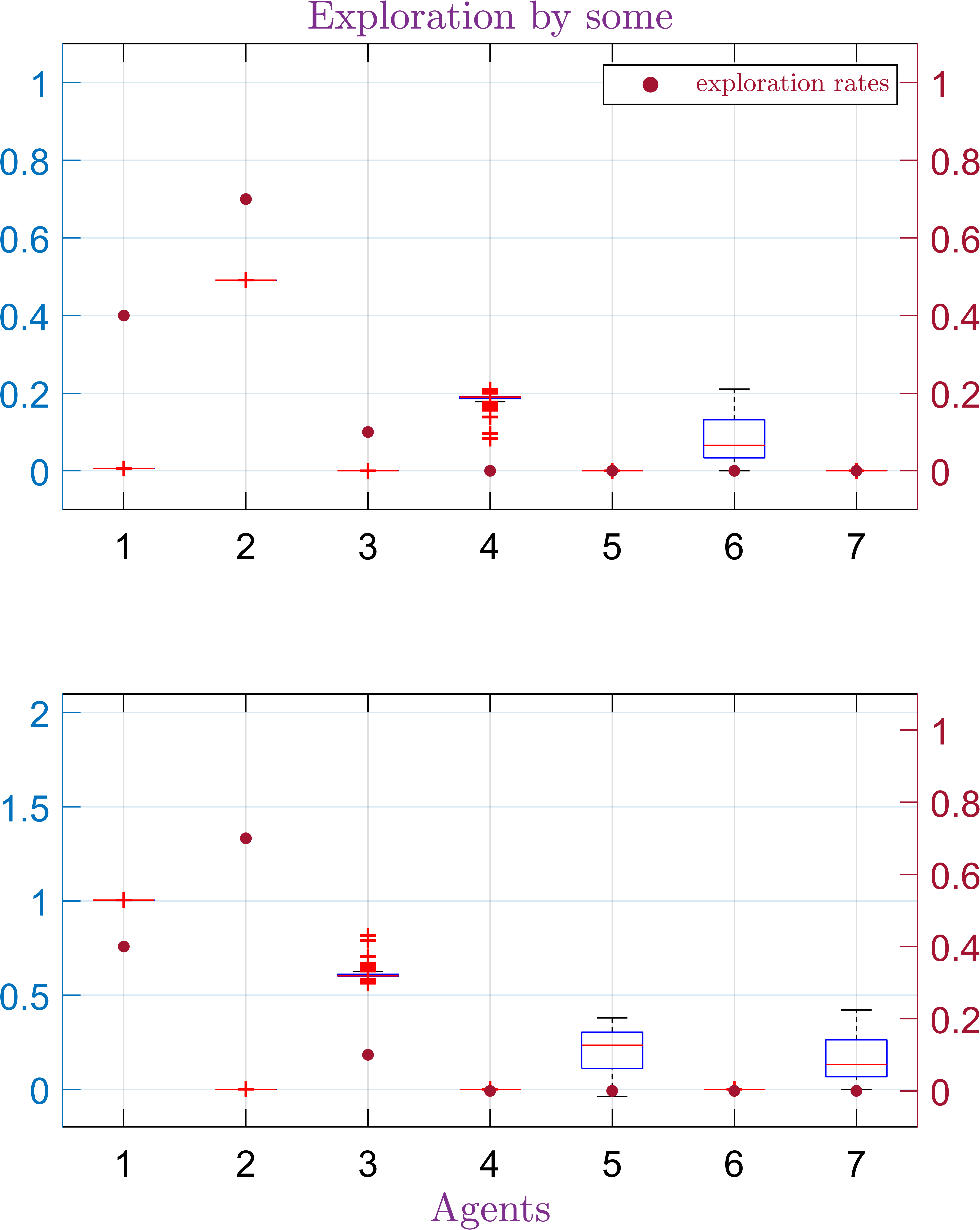}}\hspace{0.2cm}
    \includegraphics[width=0.32\linewidth,valign=b]{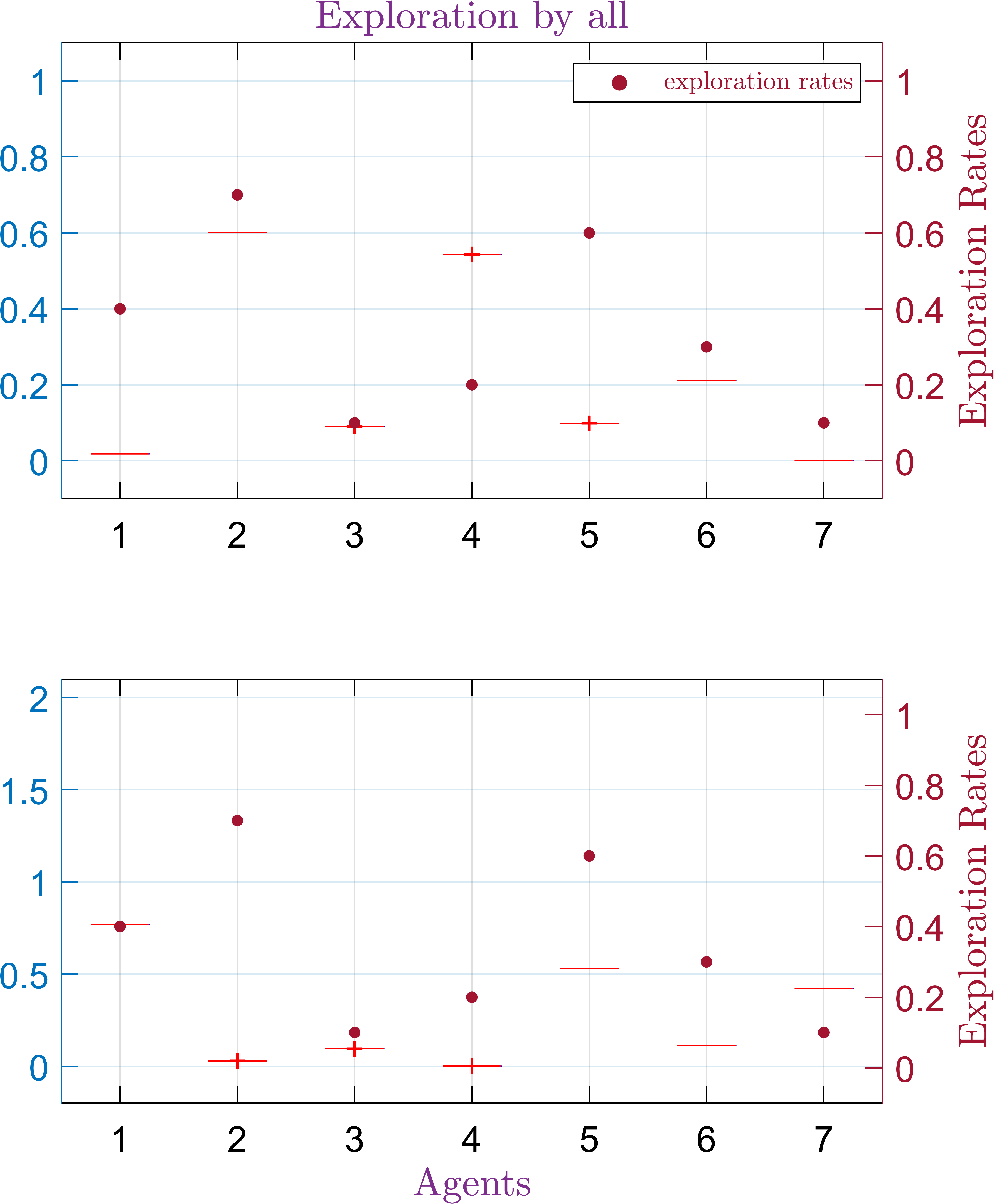}
    \caption{Different rates of exploration in an instance of the match-mismatch ZSPG with $7$ agents. The boxplots show summery statistics from $100$ runs. No exploration leads to multiple equilibria (first column), exploration by all agents leads to a unique QRE (last column), whereas equilibrium by some agents is not enough to ensure convergence to a unique outcome which shows the tightness of \Cref{thm:main} (middle column).\vspace*{-0.2cm}}
    \label{fig:summary}
\end{figure}

\section{Conclusions}
\label{sec:conclusions}
In this paper, we studied a commonly used smooth variant of the Q-learning dynamics in multi-agent competitive systems. Given that each agent's strategy includes at least some amount of exploration, we showed that the dynamics always converge to a unique Quantal Response Equilibrium (QRE). In particular, the convergence of Q-learning in this (competitive) setting is a remarkably robust phenomenon, occurring regardless of the number of agents or degree of competition, and without any need for parameter fine-tuning.\par
Our theoretical results have important implications from an algorithmic perspective. In contrast to two-agent zero-sum games 
in which agents have the same reward in all equilibria, multi-agent systems may exhibit a multiplicity of equilibria in which each agent may have different rewards. This inherent complexity within multi-agent systems leads to the non-trivial, and currently open problem, of equilibrium selection, which has previously been studied only within \emph{cooperative} settings. The bounds we prove on the speed of convergence of the Q-learning dynamics to the QRE provide the necessary theoretical guarantees for a simple algorithm to select (approximate) Nash equilibria, in turn implying that the problem of equilibrium selection is tractable within competitive multi-agent settings as well.  Our results complement previous advances on the convergence of Q-learning and equilibrium selection in cooperative settings and provide a promising starting point for exploring general multi-agent settings which combine elements from both cooperative and competitive systems.

\section*{Acknowledgements}
For this research, Stefanos Leonardos and Georgios Piliouras were supported in part by NRF2019-NRF-ANR095 ALIAS grant, grant PIE-SGP-AI-2018-01, NRF 2018 Fellowship NRF-NRFF2018-07,  AME Programmatic Fund (Grant No. A20H6b0151) from the Agency for Science, Technology and Research (A*STAR) and the National Research Foundation, Singapore under its AI Singapore Program (AISG Award No: AISG2-RP-2020-016). Kelly Spendlove was supported by EPSRC grant EP/R018472/1.

\bibliographystyle{plain}
\bibliography{bib_qlearning,refer}

\appendix
\section{Omitted Proofs and Materials: \texorpdfstring{\Cref{sec:qlearning}}{ppp}}\label{app:derivation}

We first restate and prove \Cref{lem:qcontinuous}.

{\bf\Cref{lem:qcontinuous}.}
{\em 
Consider the Q-learning updates
\begin{equation}\label{eq:appendix:qupdates}
Q_i(n+1)=Q_i(n)+\alpha[r_i(n)-Q_i(n)]
\end{equation}
where $n\ge0$ are discrete time steps and assume that $Q_i(0)=0$ for all $i\in A$. Then, in continuous time, the updates are given by
\[Q_i(t)=\alpha \int_{0}^te^{-\alpha s}r_i(t-s) ds =\alpha e^{-\alpha t}\int_{0}^te^{\alpha s}r_i(s) \mathop{ds}, \quad \text{for any } t>0.\]}
\begin{proof}
We will provide two proofs of the above statement. For the first, we solve the recursion in equation \eqref{eq:appendix:qupdates} and then use the approximation of the exponential function by the geometric sum. Specifically, we have that
\begin{align*}
Q_i(n+1)&=\alpha r_i(n)+(1-\alpha)Q_i(n)\\
&=\alpha r_i(n)+(1-\alpha)[\alpha r_i(n-1)+(1-\alpha)Q_i(n-1)]\\
&=\alpha r_i(n)+(1-\alpha)\alpha r_i(n-1)+(1-\alpha)^2Q_i(n-1)\\
&=\alpha r_i(n)+(1-\alpha)\alpha r_i(n-1)+(1-\alpha)^2[\alpha r_i(n-2)+(1-\alpha)Q_i(n-2)]\\
&= \dots\\
&=(1-\alpha)^{n+1}Q_i(0)+\sum_{k=0}^n \alpha (1-\alpha)^k r_i(n-k)\\
&=\sum_{k=0}^n \alpha (1-\alpha)^k r_i(n-k)
\end{align*}
where in the last equation we did a change of variables in the summation since $Q_i(0)=0$ by assumption. By taking continuous time steps (i.e., if instead of $k\to k+1$, we consider $k\to \Delta k$ with $\Delta k\to 0$), then the above becomes 
\[Q(t)=\alpha\int_{0}^t e^{-as}r_i(t-s)ds\]
as claimed. By a change of variables, we also obtain the second expression in the Lemma's statement. The second way to prove this statement is by considering directly a continuous time version of \eqref{eq:appendix:qupdates}. In this case, we have that 
\[\dot Q_i(s)=\alpha [r_i(s)-Q_i(s)],\]
where $\dot Q_i(s)$ denotes the time derivative of $Q_i(s)$, i.e., $d Q_i(s)/ds$. This is a first order non homogeneous (in the constant term) linear differential equation of the form
\[\dot Q_i(s)+\alpha Q_i(s)=\alpha r_i(s).\]
By multiplying both sides with the integrating factor $e^{\alpha s}$, we obtain that
\[e^{\alpha s}[\dot Q_i(s)+\alpha Q_i(s)]=\alpha e^{\alpha s}r_i(s) \implies \frac{d}{ds}(e^{\alpha s}Q_i(s))=\alpha e^{\alpha s}r_i(s).\]
We can now integrate both sides of the last expression from $0$ to $t$ to obtain that
\[e^{\alpha t}Q_i(t)=\alpha \int_{0}^te^{\alpha s}r_i(s)ds \implies  Q_i(t)=\alpha e^{-\alpha t}\int_{0}^te^{\alpha s}r_i(s)ds\]
as claimed, where the integration constant disappears due to the boundary condition $Q_i(0)=0$.
\end{proof}

\subsection{Derivation of Q-learning dynamics}
We next provide two formal derivations of the Q-learning dynamics of equation \eqref{eq:kdynamics}. Similar calculations can be found in \cite{Tuy03,Sat03,Sat05,Kia12} and \cite{Gal13,Pan17}.

\begin{proposition}[Derivation of Q-learning Dynamics]\label{prop:main}
Assume that players select their actions according to a Boltzmann (or softmax) distribution with parameter $T$, i.e.,
\begin{equation}\label{eq:boltzmann}
x_i(t)=\frac{\exp{\<Q_{i}(t)/T\>}}{\sum_{j\in A}\exp{\<Q_{j}(t)/T\>}}\,\quad \text{for any } t>0,
\end{equation}
with $x_i(0)$ is initialized arbitrarily for all $i\in A$, where $Q_i(t)$ is the Q-value of action $i$ at time $t$. Then, the evolution of the action probabilities is governed by the Q-learning dynamics
\begin{equation}\label{eq:qlear}
\dot x_i(t)=x_i(t)\lt r_i(t)-\sum_{j\in A}r_j(t)x_j(t)-T\<\ln {x_i(t)}-\sum_{j\in A}x_j(t)\ln{x_j(t)}\> \rt.
\end{equation}
\end{proposition}

\begin{proof}
The time derivative of $x_i(t)$ is equal to
\begin{align*}
\dot x_i(t)&
=\frac{\exp{\<Q_i(t)/T\>}}{T\sum_{j\in A}\exp{\<Q_{j}(t)/T\>}}\left[\dot Q_i\<t\>-\frac{\sum_{j\in A}\dot Q_j(t)\exp{\<Q_{j}(t)/T\>}}{\sum_{j\in A}\exp{\<Q_{j}(t)/T\>}}\right]\\&
=\frac1Tx_i(t)\left[\dot Q_i\<t\>-\sum_{j\in A}\<\frac{\exp{\<Q_{j}(t)/T\>}}{\sum_{j\in A}\exp{\<Q_{j}(t)/T\>}}\>\dot Q_j(t) \right]\\&
=\frac1Tx_i(t)\lt \dot Q_i\<t\>-\sum_{j\in A}x_j(t)\dot Q_j(t)\rt.
\end{align*} 
Hence, using that $\dot Q_i(t)=\alpha [r_i(t)-Q_i(t)]$, the above yields
\begin{align}\label{eq:star}
\dot x_i(t)&=\frac{\alpha}Tx_i(t)\lt r_i(t)-Q_i(t)-\sum_{j\in A}x_j(t)[r_j(t)-Q_j(t)]\rt\nonumber\\[0.2cm]
&=\frac{\alpha}Tx_i(t)\lt r_i(t)-\sum_{j\in A}x_j(t)r_j(t)-\<Q_i(t)-\sum_{j\in A}x_j(t)Q_j(t)\>\rt \tag{$\ast$}.
\end{align}
Finally, taking $\ln{}$ in equation \eqref{eq:boltzmann}, and solving for $Q_i(t)$, we have that  
\[Q_i(t)=T\lt \ln{x_i(t)}+\ln{\<\sum_{j\in A}\exp{(Q_j(t)/T)}\>}\rt\]
In the above expression, we may set $C:=\ln{\<\sum_{j\in A}\exp{(Q_j(t)/T)}\>}$ since this term is the same for all $i\in A$. Thus, we have that 
\begin{align*}
Q_i(t)-\sum_{j\in A}x_j(t)Q_j(t)&=T\<\ln{x_i(t)}+C\>-\sum_{j\in A}x_j(t)T\<\ln{x_i(t)}+C\>\\
&=T\< \ln{x_i(t)}-\sum_{j\in A}x_j(t)\ln{x_j(t)}\> +T\< C-\sum_{j\in A}x_j(t)C\>\\
&=T\< \ln{x_i(t)}-\sum_{j\in A}x_j(t)\ln{x_j(t)}\>
\end{align*}
since $\sum_{j\in A}x_j(t)=1$ and hence, the terms $C$ and $\sum_{j\in A}x_j(t)C$ cancel out. Substituting the last expression in equation \eqref{eq:star}, we obtain the desired solution, namely
\[
\dot x_i(t)=\frac{\alpha}Tx_i(t)\lt r_i(t)-\sum_{j\in A}x_j(t)r_j(t)-T\< \ln{x_i(t)}-\sum_{j\in A}x_j(t)\ln{x_j(t)}\>\rt.
\]
Rescaling time by $t\to \alpha t/T$, we obtain the solution.
\end{proof}

\paragraph{An Alternative Derivation} (see also \cite{Gal13,Pan17}). Another way to obtain the Q-learning dynamics in equation \eqref{eq:qlear} is by a direct substitution of \eqref{eq:appendix:qupdates} in \eqref{eq:boltzmann}. In this case, we have that
\begin{align*}
x_i(t+1)&=\frac{\exp{\<Q_{i}(t+1)/T\>}}{\sum_{j\in A}\exp{\<Q_{j}(t+1)/T\>}}\\[0.2cm]
&=\frac{\exp{\<\<\alpha r_i(t+1)+(1-\alpha)Q_i(t)\>/T\>}}{\sum_{j\in A}\exp{\<\<\alpha r_j(t+1)+(1-\alpha)Q_j(t)\>/T\>}}\\[0.2cm]
&=\frac{\exp{\<(1-\alpha)Q_i(t)/T\>}\cdot\exp{\<\alpha r_i(t+1)/T\>}}{\sum_{j\in A}\exp{\<(1-\alpha)Q_j(t)/T\>}\cdot\exp{\<\alpha r_j(t+1)/T\>}}\\[0.2cm]
&=\frac{\<\exp{\<Q_i(t)/T\>}\>^{(1-\alpha)}\cdot\exp{\<\alpha r_i(t+1)/T\>}}{\sum_{j\in A} \<\exp{\<Q_j(t)/T\>}\>^{(1-\alpha)}\cdot\exp{\<\alpha r_j(t+1)/T\>}}\\&=\frac{x_i(t)^{(1-\alpha)}\cdot\exp{\<\alpha r_i(t+1)/T\>}}{\sum_{j\in A}x_j(t)^{(1-\alpha)}\cdot\exp{\<\alpha r_j(t+1)/T\>}}
\end{align*}
where the last equality is obtained by dividing both numerator and denominator with the normalizing constant $\sum_{k\in A}\<\exp{\<Q_k(t)/T\>}\>^{(1-\alpha)}$. By taking $\ln{}$ of both sides in the previous equation, we then have that
\[\ln{x_i(t+1)}=\ln{x_i(t)}-\alpha\ln{x_i(t)}+\frac{\alpha}Tr_i(t+1)-\ln{\<\sum_{j\in A}x_j(t)^{(1-\alpha)}\cdot\exp{\<\alpha r_j(t+1)/T\>}\>}\]
or equivalently
\begin{align*}
\ln{x_i(t+1)}-\ln{x_i(t)}&=\frac{\alpha}Tr_i(t+1)-\alpha\ln{x_i(t)}-\ln{C}
\end{align*}
where $C:=\sum_{j\in A}x_j(t)^{(1-\alpha)}\cdot\exp{\<\alpha r_j(t+1)/T\>}$ is the denominator of the previous expression which is the same for all $i\in A$. Thus, in continuous time, the above equation becomes 
\[\frac{d}{dt}\ln{x_i(t)}=\frac{\alpha}{T}r_i(t)-\alpha \ln{x_i(t)} -\ln{C}\]
which yields
\begin{align}\label{eq:astast}
\dot x_i(t)=\frac{\alpha}{T}x_i(t)\lt r_i(t)-T\ln{x_i(t)} -\frac{T}{\alpha}\ln{C}\rt.\tag{$\ast\ast$}
\end{align}
To determine $\ln{C}$, note that $\sum_{i\in A}\dot x_i(t)=0$ since $\sum_{i\in A}x_i(t)=1$ for all $t$ (i.e., the sum of the $x_i$'s
remains constant, and equal to 1, at all times $t\ge0$). Thus, summing over all $i\in A$, we obtain 
\begin{align*}
0&=\sum_{i\in A}x_i(t)r_i(t)-T\sum_{i\in A}x_i(t)\ln{x_i(t)} -\frac{T}{\alpha}\ln{C}\sum_{i\in A}x_i(t)\\
&=\sum_{i\in A}x_i(t)r_i(t)-T\sum_{i\in A}x_i(t)\ln{x_i(t)} -\frac{T}{\alpha}\ln{C}
\end{align*}
or equivalently
\[\frac{T}{\alpha}\ln{C}=\sum_{i\in A}x_i(t)r_i(t)-T\sum_{i\in A}x_i(t)\ln{x_i(t)}\]
Substituting this expression back in equation \eqref{eq:astast}, we obtain 
\begin{align*}
\dot x_i(t)&=\frac{\alpha}{T}x_i(t)\lt r_i(t)-T\ln{x_i(t)} -\sum_{j\in A}x_j(t)r_j(t)+T\sum_{j\in A}x_j(t)\ln{x_j(t)}\rt\\
&=\frac{\alpha}{T}x_i(t)\lt r_i(t)-\sum_{j\in A}x_j(t)r_j(t)-T\<\ln{x_i(t)}-\sum_{i\in A}x_j(t)\ln{x_j(t)}\>\rt,
\end{align*}
which after rescaling time to $t\to\alpha t/T$ is precisely the expression of the Q-learning dynamics in equation \eqref{eq:qlear}.

\paragraph{Relation to Experience Weighted Attraction (EWA) Learning} (see also \cite{Pan17}). In Experience Weighted Learning (EWA), the \emph{attractions} (the equivalent of Q-values), $Q_i(t)$ of each action $i\in A$ are updated according to the following scheme
\begin{align*}
Q_i(t+1)&=\frac{(1-\alpha) N(t)Q_i(t)+[\delta+(1-\delta)I(i,s_i(t))]r_i(t)}{N(t+1)},\\
N(t+1)&=\rho N(t)+1,
\end{align*}
where $I(i,s(t))=1$ if $i=s(t)$ and $0$ otherwise. Here $s(t)$ denotes the action taken by the agent at time $t$. The variables $N(t)$ and $Q_i(t)$ are initialized arbitrarily at $t=0$, but typically, they are set to be $0$ at $t=0$. For $\rho=0$, i.e., in the case in which previous experience does not reduce the impact of current rewards, the above system becomes
\begin{align*}
Q_i(t+1)&=(1-\alpha)Q_i(t)+\begin{cases}\phantom{\delta}r_i(t), & \text{if } s(t)=i, \\ \delta r_i(t), & \text{otherwise} \end{cases} 
\end{align*}
Thus, for $\delta=1$, the EWA update rule becomes equal to the Q-learning updates, up to a constant $\alpha$ in the rewards.

\subsection{Q-learning and Quantal Response Equilibria}

We next restate and prove \Cref{thm:preliminary}.

{\bf \Cref{thm:preliminary}.}
{\em 
Let $\Gamma$ be an arbitrary game, with positive exploration rates $T_k$ and consider the associated Q-learning dynamics
\[\dot\xk=\xk\lt r_{ki}(\x_{-k})-\x_k^\top\rk -T_k\< \ln{(\xk)}-\x_k^\top\ln{(\x_k)} \> \rt,\;\; i\in S_k, k\in V.\] 
The interior fixed points, $\p=(p_k)_{k\in V}$ of the Q-learning dynamics are the solutions of the system  
\begin{equation}\label{eq:appendix:qre}
p_{ki}=\frac{\exp{\<r_{ki}(\p_{-k})/T_k\>}}{\sum_{j\in S_k}\exp{\<r_{kj}(\p_{-k})/T_k\>}}, \quad \text{for all } i\in S_k.
\end{equation}
Such fixed points always exists and coincide with the Quantal Response Equilibria (QRE) of $\Gamma$. Given any such fixed point $\p$, we have, for all $\x_k\in \Delta_k$ and for all $k\in V$, that
\begin{equation}\label{eq:appendix:eqprop}
(\x_k-\p_k)^\top \lt r_{k}(\p_{-k})-T_k \ln{(\p_k)}\rt=0.
\end{equation}
}

\begin{proof}
Solving equation \eqref{eq:kqre} for $\ln{q_{ki}}$ and applying the exponential function on both sides of the resulting equation yields that
\begin{align*}
p_{ki}&=\exp{\<\frac{r_{ki}(\p_{-k})-\p_k^\top \rp +T_k\p_k^\top\ln{(\p_k)}}{T_k}\>}
\\&=\exp{\<\frac{r_{ki}(\p_{-k})}{T_k}\>}\cdot\exp{\<\frac{-\p_k^\top \rp +T_k\p_k^\top\ln{(\p_k)}}{T_k}\>}.
\end{align*}
The second term in the last equation, i.e., $\exp{\<\frac{-\p_k^\top \rp +T_k\p_k^\top\ln{(\p_k)}}{T_k}\>}$, is the same (constant) for all $i \in S_k$. Thus, denoting this term by $Z$, we have that 
\[p_{ki}=\exp{\<\frac{r_{ki}(\p_{-k})}{T_k}\>}\cdot Z, \quad \text{for all} i\in S_k.\]
Since $\p_k$ lies in $\Delta_k$, it must be the case that $\sum_{j\in S_k}p_{ki}=1$, which implies that 
\[\sum_{i\in S_k}p_{ki}=\sum_{i\in S_k}\exp{\<\frac{r_{ki}(\p_{-k})}{T_k}\>}\cdot Z=1\implies Z=\<\sum_{i\in S_k}\exp{\<\frac{r_{ki}(\p_{-k})}{T_k}\>}\>^{-1}.\]
Substituting back in the expression for $p_{ki}$, we obtain that 
\[p_{ki}=\frac{\exp{\<r_{ki}(\p_{-k})/T_k\>}}{\sum_{j\in S_k}\exp{\<r_{kj}(\p_{-k})T_k\>}},\]
as claimed in equation \eqref{eq:appendix:qre}. Existence follows from the application of Brouwer's fixed point theorem on the continuous map defined by the previous in $\Delta_k$, see \cite{Mck95}. To obtain equation \eqref{eq:appendix:eqprop}, we observe that for each $x_k\in \Delta_k$ and each $k\in V$, equation \eqref{eq:kqre} implies that
\[\x_k^\top r_{k}(\p_{-k})=\p_k^\top \rp -T_k (\p_k-\x_k)^\top\ln{(\p_k)}, \;\; \text{for all } k\in V.\]
or equivalently that
\[(\x_k-\p_k)^\top \lt r_{k}(\p_{-k})-T_k \ln{(\p_k)}\rt=0,\]
for all $k\in V$ as claimed in equation \eqref{eq:appendix:eqprop}.
\end{proof}

\section{Omitted Proofs and Materials: \texorpdfstring{\Cref{sec:convergence}}{ooo}}

In this section, we develop the necessary technical framework for the proof of \Cref{thm:main}. \Cref{thm:main} is based upon two critical lemmas and properties of both KL divergence and rescaled zero-sum polymatrix games. \par

First, recall that in KL-divergence is not symmetric, i.e., it need not hold that $\KL(\p\|\x(t))=\KL(\x(t)\|\p)$.  However, KL-divergence does obey the following property.

\begin{property}\label{prop:symmetric}
Let $k \in V$ and let $\p_k,\x_k$ be interior points of $\Delta_k$. Then, it holds that 
\[\KL(\p_k\|\x_k)+\KL(\x_k\|\p_k)=(\x_k-\p_k)^\top\lt \ln{(\x_k)}-\ln{(\p_k)}\rt.\]
\end{property}
\begin{proof}
It is immediate to check that 
\begin{align*}
(\x_k-\p_k)^\top\lt \ln{(\x_k)}-\ln{(\p_k)}\rt&=\x_k^\top\ln{\<\frac{\x_k}{\p_k}\>}-\p_k^\top\ln{\<\frac{\x_k}{\p_k}\>}\\
&=\x_k^\top\ln{\<\frac{\x_k}{\p_k}\>}+\p_k^\top\ln{\<\frac{\p_k}{\x_k}\>}\\[0.15cm]
&=\KL(\x_k\|\p_k)+\KL(\p_k\|\x_k),
\end{align*}
as claimed.
\end{proof}

{\bf\Cref{lem:tderivative}.}
{\em 
Let $k\in V$. The time-derivative of the $\KL$-divergence between the $k$-th component, $\p_k \in \Delta_k$, of a QRE $\p\in \Delta$ of $\Gamma$, and the $k$-th component, $\x_k(t)\in \Delta_k$ of a system trajectory with $\x(0)$ an interior point, is given by
\begin{equation}\label{eq:appendix:timederivative}
\frac{d}{dt}\KL(\p_k\|\x_k(t))=(\x_k-\p_k)^\top \lt \rk-r_{k}(\p_{-k})\rt-T_k\lt \KL(\p_k\|\x_k)+\KL(\x_k\|\p_k) \rt .
\end{equation}
}
\begin{proof}
The time derivative of the term $\KL(\p_k\|\x_k)$ for $k\in V$ can be calculated as follows (note that after the first line, we omit the dependence of $\x$ on $t$ to simplify notation)
\begin{align*}
\frac{d}{dt}\KL(\p_k\|\x_k(t))&= -\sum_{i\in S_k} p_{ki}\frac{d}{dt}(\ln{(\xk(t))})=-\sum_{i\in S_k} p_{ki}\frac{\dot\xk(t)}{\xk(t)}\\
&= -\sum_{i\in S_k} p_{ki}\lt r_{ki}(\x_{-k})-\x_k^\top \rk+T_k\< -\ln{(\xk)}+\x_k^\top\ln{(\x_k)}\> \rt\\
&=-\lt \p_k^\top \rk-\x_k^\top \rk+T_k\<-\p_k^\top \ln{(\x_k)}+\x_k^\top \ln{(\x_k)}\>\rt \\[0.2cm]
&=\phantom{-\,} (\x_k-\p_k)^\top \rk -T_k(\x_k-\p_k)^\top\ln{(\x_k)}\\[0.2cm]
&=\phantom{-\,} (\x_k-\p_k)^\top \lt \rk -T_k\ln{(\x_k)}\rt.
\end{align*}
Using equation \eqref{eq:appendix:eqprop}, i.e., that 
\[(\x_k-\p_k)^\top \lt r_k(\p_{-k}) -T_k\ln{(\p_k)}\rt=0,\]
we can write the time derivative of $\KL$ as follows
\begin{align*}
\frac{d}{dt}\KL(\p_k\|\x_k(t))&=(\x_k-\p_k)^\top \lt \rk -T_k\ln{(\x_k)}\rt-(\x_k-\p_k)^\top \lt r_k(\p_{-k}) -T_k\ln{(\p_k)}\rt\\
&=(\x_k-\p_k)^\top \lt \rk-r_k(\p_{-k})\rt-T_k(\x_k-\p_k)^\top\lt\ln{\x_k}-\ln{\p_k}\rt\\
&=(\x_k-\p_k)^\top \lt \rk-r_k(\p_{-k})\rt-T_k\lt \KL(\p_k\|\x_k)+\KL(\x_k\|\p_k)\rt
\end{align*}
where the last equation holds due to \Cref{prop:symmetric}. This concludes the proof of the Lemma.
\end{proof}

{\bf \Cref{lem:summation}.}
{\em 
Let $\p=(\p_k)_{k\in V}$ be a QRE of $\Gamma$ and let $\x=(\x_k)_{k\in V}$. Then, it holds that
\begin{equation}\label{eq:appendix:summation}
\sum_{k\in V}w_k\lt \x_k^\top \rp+\p_k^\top\rk\rt=0.
\end{equation}
}
Before proceeding with the proof of \Cref{lem:summation}, note that if there only two players, i.e., if $V=\{1,2\}$, then $\p=(\p_1,\p_2)$ and $\x=(\x_1,\x_2)$, and it is rather immediate to check the validity equation \eqref{eq:appendix:summation}, since
\begin{align*}
\sum_{k=1,2}w_k[ \x_k^\top \rp&+\p_k^\top\rk]=\\&=w_1\lt\x_1^\top r_1(\p_2)+\p_1^\top r_1(\x_2)\rt+w_2\lt\x_2^\top r_2(\p_1)+w_2\p_2^\top r_2(\x_1)\rt\\
& = \lt w_1\x_1^\top r_1(\p_2)+w_2\p_2^\top r_2(\x_1)\rt+\lt w_2\x_2^\top r_2(\p_1)+w_2\p_1^\top r_1(\x_2)\rt\\[0.2cm]
& = \underbrace{\sum_{k=1,2} w_k u_k(\x_1,\p_2)}_{=\,0 \text{ by equation \eqref{eq:zsprop}}} +\underbrace{\sum_{k=1,2}w_ku_k(\p_1,\x_2)}_{=\,0 \text{ by equation \eqref{eq:zsprop}}}=0.
\end{align*}
In particular, the last equation holds because the summations are over all payoffs in the strategy profiles $(\x_1,\p_2)$ and $(\p_1,\x_2)$. When we have more than two players, this argument does not hold, since the summation is over different strategy profiles. However, it still holds that the summation is equal to zero. To show this, we will need to apply the following property that has been established in \cite{Cai16} for zero-sum polymatrix games to the case of weighted zero-sum polymatrix games (the extension is rather straightforward as we show below). To state the property in the general case, we will use the following definition. 

\begin{definition}[$w_k$-Payoff equivalence]
Consider two arbitrary games $\Gamma=\<(V,E),(S_k,u_k)_{k\in V}\>$ and $\Gamma'=\<(V,E),(S_k,u_k')_{k\in V}\>$. We will say that $\Gamma$ is $w_k$-payoff equivalent to $\Gamma'$ if there exist positive constants $w_k, k\in V$ so that 
\[u_k(\x)=w_ku_k'(\x),\;\; \text{for all } x\in \Delta.\]
\end{definition}

\begin{property}[Payoff equivalent transformation \cite{Cai16}.]\label{prop:transformation}
Let $\Gamma=\<(V,E), \<S_k,w_k\>_{k\in V},\<\A_{kl}\>_{[k,l]\in E}\>$ be a rescaled zero-sum polymatrix game. Then, $\Gamma$ is $1/w_k$-payoff equivalent to a pairwise constant-sum (unweighted) polymatrix game $\hat \Gamma=\<(V,E), \<S_k\>_{k\in V},\<\hat\A_{kl}\>_{[k,l]\in E}\>$, i.e., a game in which every two-player game $[k,l]\in E$ is constant-sum and all these constants sum up to zero. Specifically, for all $k,l\in V$ with $[k,l]\in E$, there exist payoff matrices $\hat\A_{kl}=\<\hat a_{kl}(s_k,s_l)\>_{\{s_k\in S_k,s_l\in S_l\}}$ and constants $c_{kl} \in \mathbb R$, so that
\begin{align}\label{eq:appendix:trans1}
\hat a_{kl}(s_k,s_l)+\hat a_{lk}(s_l,s_k)&=c_{kl}, \;\; \text{for all } s_k\in S_k, s_l\in S_l,
\end{align}
with 
\begin{equation}\label{eq:appendix:trans3}
\sum_{[k,l]\in E} c_{kl}=0,
\end{equation}
and 
\begin{equation}\label{eq:appendix:trans2}
\hat r_{ki}(\x_{-k}):=\sum_{[k,l]\in E}\{\hat \A_{kl}\x_l\}_i=\sum_{[k,l]\in E}w_k\{\A_{kl}\x_l\}_i=w_kr_{ki}(\x_{-k}),
\end{equation}
for all $k\in V$ and all $i\in S_k$, i.e., for every pure (and hence, also for every mixed) strategy profile, the payoff of each player $k\in V$ in $\hat \Gamma$ is equal to $1/w_k$ their payoff in the original game $\Gamma$.
\end{property}

\begin{proof}
As mentioned above, all claims of \Cref{prop:transformation} have been established for (unweighted) zero-sum polymatrix games in \cite{Cai16}. Thus, it remains to show that the proof extends to the weighted case. To see this, consider a weighted zero-sum polymatrix game $\Gamma$ with weights $\w=(w_k)_{k\in V}$, payoff matrices $\<\A_{kl}\>_{[k,l]\in E}$ and utilities $u_k,k\in V$ as in equation \eqref{eq:utility} and define the transformed game $\Gamma^\B$ with payoff matrices $\<\B_{kl}\>_{[k,l]\in E}$ given by
\[\B_{kl}:=w_k\A_{kl}, \quad \text{for all } k \in V, [k,l]\in E,\]
and utilities $u_k^\B, k\in V$.
Then, it follows from equation \eqref{eq:utility} that $u_k(\x)=\frac{1}{w_k}u_k^\B(\x)$ for all $\x\in \Delta$ and $k\in V$ since
\[u_k^\B(\x)=\sum_{[k,l]\in E} \x_k^\top \B_{kl} \x_l =\sum_{[k,l]\in E}\x_k^\top \<w_k\A_{kl}\>\x_l=w_k\sum_{[k,l]\in E}\x_k^\top \A_{kl}\x_l=w_ku_k(\x).\]
Moreover, $\Gamma^\B$ is an \emph{unweighted} zero-sum game since
\[\sum_{k\in V}u_k^\B(\x)=\sum_{k\in V}w_ku_k(\x)=0,\]
where the last equality follows from the fact that $\Gamma$ is a weighted zero-sum polymatrix game with weights $(w_k)_{k\in V}$. Thus, transforming $\Gamma$ to $\Gamma^\B$ and then applying the transformation of \cite{Cai16} on $\Gamma^\B$ to obtain the payoff equivalent (to $\Gamma^\B$) game $\hat\Gamma$ with payoff matrices $\hat \A_{kl}, [k,l]\in E$ yields the result, i.e., the $1/w_k$-payoff equivalence between the original weighted zero-sum polymatrix game $\Gamma$ and the pairwise constant-sum, unweighted polymatrix game $\hat\Gamma$. 
\end{proof}

Using \Cref{prop:transformation}, we can now prove \Cref{lem:summation} for general $k\ge2$.
\begin{proof}[Proof of \Cref{lem:summation}]
Equation \eqref{eq:appendix:trans2} in \Cref{prop:transformation} implies that
\begin{align*}
\sum_{k\in V}w_k\lt \x_k^\top \rp+\p_k^\top\rk\rt& \overset{\eqref{eq:appendix:trans2}}{=}\sum_{k\in V}w_k\lt \x_k^\top \<\frac{1}{w_k}\hat r_k(\x_{-p})\>+\p_k^\top\<\frac{1}{w_k}\hat r_k(\x_{-k})\>\rt\\
\overset{\phantom{(18)}}{=} &\sum_{k\in V}\sum_{[k,l]\in E} \lt \x_k^\top \hat \A_{kl}\p_l+\p_k^\top \hat \A_{kl}\x_l\rt\\
\overset{\phantom{(18)}}{=} &\sum_{[k,l]\in E} \lt \x_k^\top \hat \A_{kl}\p_l+\p_k^\top \hat \A_{kl}\x_l +\x_l^\top \hat \A_{lk}\p_k+\p_l^\top \hat \A_{kl}\x_k\rt\\
\overset{\phantom{(18)}}{=} &\sum_{[k,l]\in E} \underbrace{\lt \x_k^\top \hat \A_{kl}\p_l+\p_l^\top \hat \A_{kl}\x_k\rt}_{=\, c_{kl}\text{ by equation \eqref{eq:appendix:trans1}}}+\underbrace{\lt\x_l^\top \hat \A_{lk}\p_k +\p_k^\top \hat \A_{kl}\x_l \rt}_{=\, c_{kl}\text{ by equation \eqref{eq:appendix:trans1}}}\\
\overset{\phantom{(18)}}{=}& 2\sum_{[k,l]\in E}c_{kl} = 0,
\end{align*}
where the last equality holds by equation \eqref{eq:appendix:trans3}.
\end{proof}

\subsection{Proof of \texorpdfstring{\Cref{thm:main}}{yyy}}

Combining \Cref{lem:tderivative,lem:summation}, we can now restate and prove \Cref{thm:main}.

{\bf\Cref{thm:main}.}
{\em 
Let $\Gamma$ be a rescaled zero-sum polymatrix game, with positive exploration rates $T_k$.  There exists a unique QRE $\p$ such that if $\x(t)$ is any trajectory of the associated Q-learning dynamics $\dot \x = f(\x)$,  with $f_i$ is given via \eqref{eq:kdynamics}, where $\x(0)$ is an interior point, then $\x(t)$ converges to $\p$ exponentially fast.  In particular, we have that 
\begin{equation}\label{eqn:appendix:KLlyapunov}
\frac{d}{dt}\wKL(\p\|\x(t))=-\sum\nolimits_{k\in V}w_kT_k\lt \KL(\p_k\|\x_k)+\KL(\x_k\|\p_k) \rt.
\end{equation}
}
\begin{proof}

\Cref{thm:preliminary} states that there must exist some a QRE equilibrium $\p=(\p_k)$.  We first establish \eqref{eqn:appendix:KLlyapunov}, from which the remaining statements will follow. 
The time derivative of the KL-divergence between $\p$ and $\x(t)$ is given by (as in the proof of \Cref{lem:tderivative}, we again omit the dependence of $\x$ on $t$ after the first line to simplify notation)
\begin{align*}
\frac{d}{dt}&\KL(\p\|\x(t))=~\sum_{k\in V}\frac{d}{dt}w_k\KL(\p_k\|\x_k)\\
&=~\sum_{k\in V}w_k\lt (\x_k-\p_k)^\top \lt \rk-r_{k}(\p_{-k})\rt-T_k\< \KL(\p_k\|\x_k)+\KL(\x_k\|\p_k) \>\rt\\
&=~\sum_{k\in V}w_k(\x_k-\p_k)^\top \lt \rk-r_{k}(\p_{-k})\rt-\sum_{k\in V}w_kT_k\lt \KL(\p_k\|\x_k)+\KL(\x_k\|\p_k) \rt.
%
\end{align*}
The first term in the right hand side of the last equation is equal to zero, since
\begin{align*}
\sum_{k\in V}w_k(\x_k-\p_k)^\top \lt \rk-r_{k}(\p_{-k})\rt&~=~\sum_{k\in V}w_k\x_k^\top \rk+\sum_{k\in V}w_k\p_k^\top r_k(\p_{-k})\\
&~~+\sum_{k\in V}w_k\lt \x_k^\top \rp+\p_k^\top\rk\rt=0,
\end{align*}
where in the last equality, we used the zero-sum property (cf. equation \eqref{eq:zsprop}) to conclude that the terms $\sum_{k\in V}w_k\x_k^\top\rk$ and $\sum_{k\in V}w_k\p_k^\top \rp$ are equal to $0$ and equation \eqref{eq:appendix:summation} in \Cref{lem:summation} to conclude that the term $\sum_{k\in V}w_k\lt \x_k^\top \rp+\p_k^\top\rk\rt$ is also equal to $0$. Thus, 
\begin{align}\label{eq:appendix:KLDeriv}
\frac{d}{dt}\KL(\p\|\x(t))=-\sum_{k\in V}w_kT_k\lt \KL(\p_k\|\x_k)+\KL(\x_k\|\p_k) \rt ,
\end{align}
which implies that $\frac{d}{dt}\KL(\p\|\x)<0$ for all $\x\neq \p$ as $T_k>0$ for each $k$.  In other words,  $\Phi(\x):=\KL(\p\|\x)$ obeys the properties i) $\Phi(\p)=0$, ii) $\Phi(\x)> 0$ if $\x\neq \p$ and iii) $\dot \Phi(\x) <0$ for $\x\neq \p$, i.e., $\Phi$  is a Lyapunov function for the Q-learning dynamics \eqref{eq:kdynamics}. In particular, since $\KL\<\x_k \| \p_k \>>0$ as long as the distributions $\p_k,\x_k$ are not equal, then
\begin{equation}\label{eqn:phiexpstability}
\dot  \Phi(x) \leq  -\min_k T_k \cdot \Phi(x),
\end{equation}
thus the KL-divergence converges to zero exponentially fast.  
As any QRE $\p'\neq \p$  must also satisfy $\dot \Phi(\p')=\nabla \Phi\cdot f(\p') =0$ (as $\p'$ is a fixed point for $\dot \x = f(\x)$, where $f$ is given in \eqref{eq:kdynamics}), but we have that $\dot \Phi(\x)<0$ for all $\x\neq \p$, it follows that $\p$ is unique.
\end{proof}

\section{Additional Experiments}\label{app:experiments}

In the section, we present additional simulations and calculations that complement our experimental results in \Cref{exp:network} in the main paper. 

\subsection{Q-learning dynamics in the AMPs game}
We start with \Cref{fig:surface_sup} which completes the possible combinations of the two representative policies that we consider, CLR-1 and ETE, in the AMPs game (cf. \Cref{exp:network}). The plots are similar to that in \Cref{fig:surface}.
 
\begin{figure}[!tb]
    \centering
    \raisebox{-0.5cm}{\includegraphics[width=0.49\linewidth]{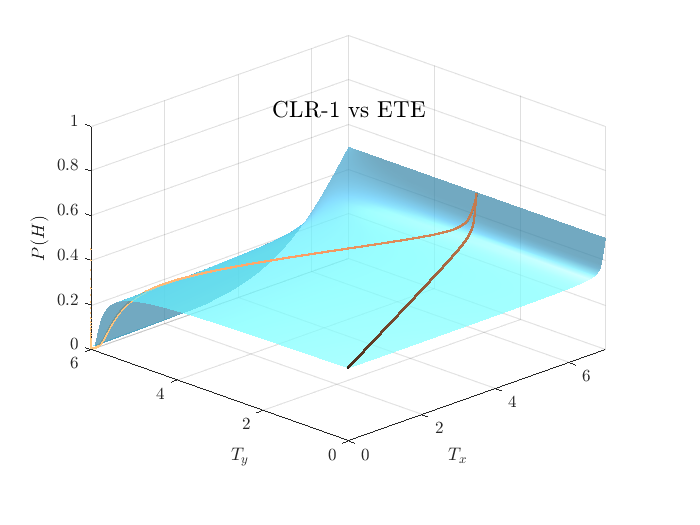}}\hfill
    \raisebox{-0.5cm}{\includegraphics[width=0.49\linewidth]{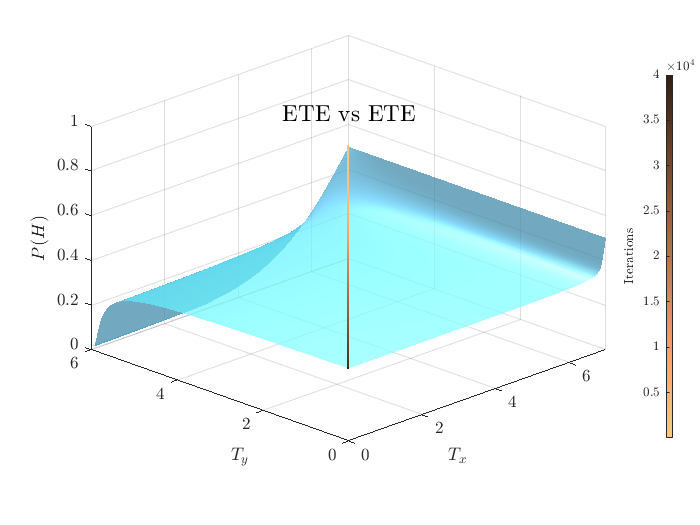}}\\
    \includegraphics[width=0.23\linewidth]{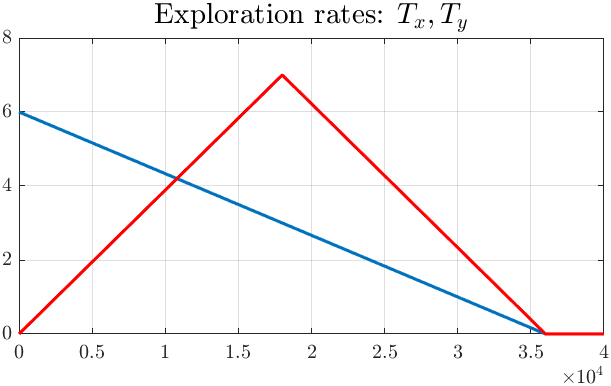}\hspace{2pt}
    \includegraphics[width=0.23\linewidth]{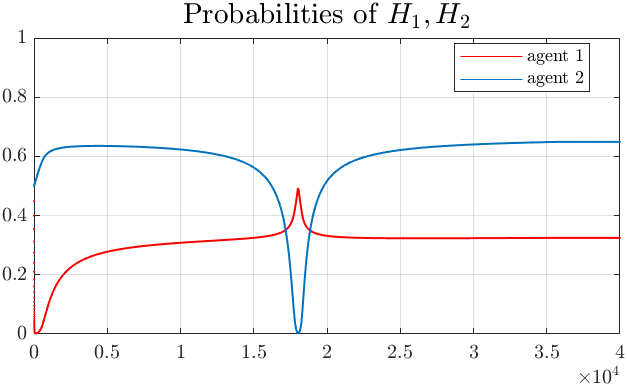}\hfill
    \includegraphics[width=0.23\linewidth]{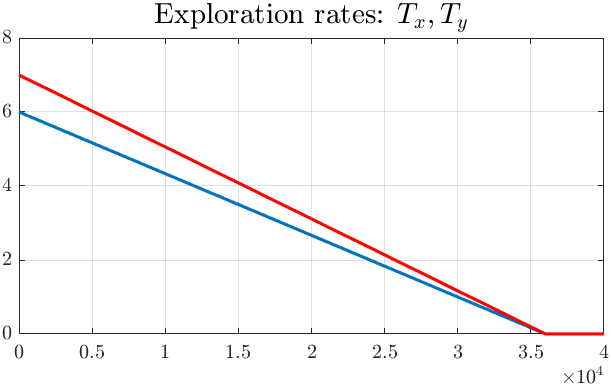}\hspace{2pt}
    \includegraphics[width=0.23\linewidth]{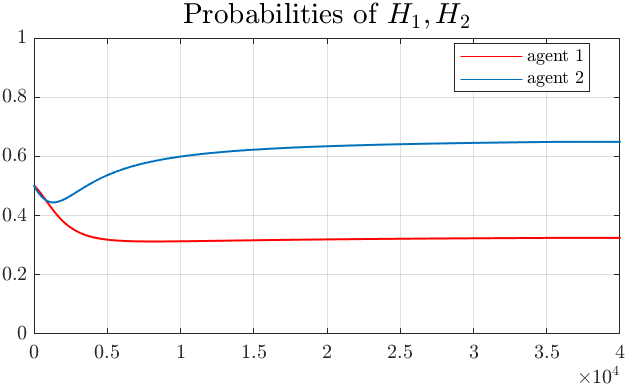}
    \caption{QRE surface and exploration paths (upper panels) for two different exploration policy profiles (lower panels) in the Asymmetric MPs game. As in \Cref{fig:surface} in the main part, the sequence of play converges to the unique QRE for any combination of exploration policies (CLR-1 and ETE). As the exploration rates decrease zero, the sequence of play converges to the unique Nash equilibrium of the game.}
    \label{fig:surface_sup}
\end{figure} 
 
 \subsection{Q-learning dynamics in two-player games with more actions}\label{sub:more_actions}

We next turn to visualizations of the Q-learning in games with two-players but more than two actions for each player. For this purpose, we consider the (symmetric) zero-sum game, Rock-Paper-Scissors (RPS) with payoff matrices given by 
\[
\A=\begin{pmatrix}\phantom{-}0 & -1 & \phantom{-}1\\ \phantom{-}1 & \phantom{-}0 & -1 \\ -1 & \phantom{-}1 & \phantom{-}0\end{pmatrix},\quad \B=-\A^{\top}.\]

The RPS game has a unique (interior) Nash equilibrium given by $(p^*,q^*)=((1/3,1/3,1/3)$, $(1/3,1/3,1/3))$. This is also the unique QRE for any positive exploration rates (since exploration favors the uniform distribution). In \Cref{fig:rps}, we visualize the Q-learning dynamics (cf. equation \eqref{eq:kdynamics}) for $T_x=0$ and various values of $T_y$. The setup and the plots are similar to that of \Cref{fig:zero}. 
\begin{figure}[!tb]
    \centering
    \includegraphics[width=0.32\linewidth]{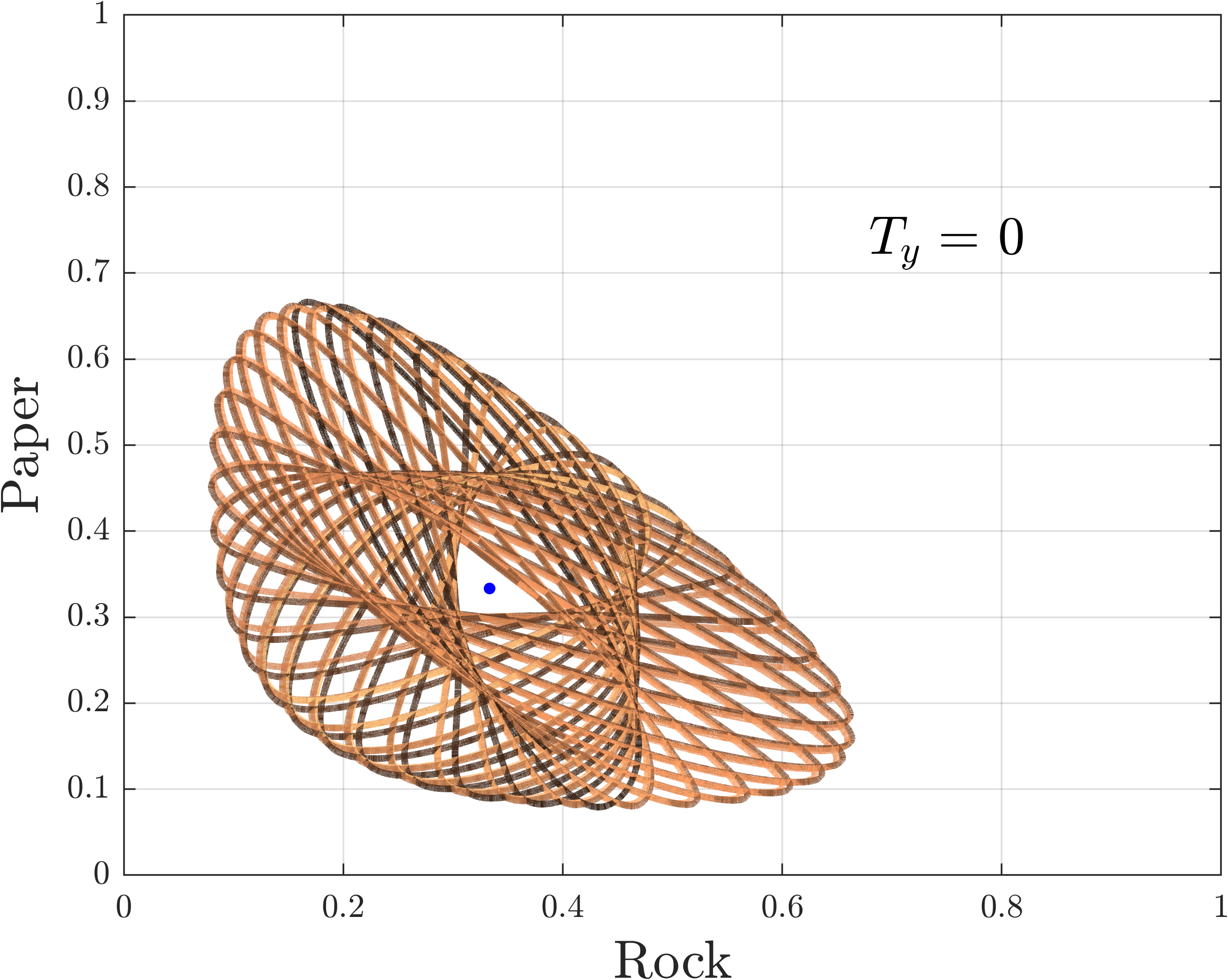}\hspace{5pt}
    \includegraphics[width=0.32\linewidth]{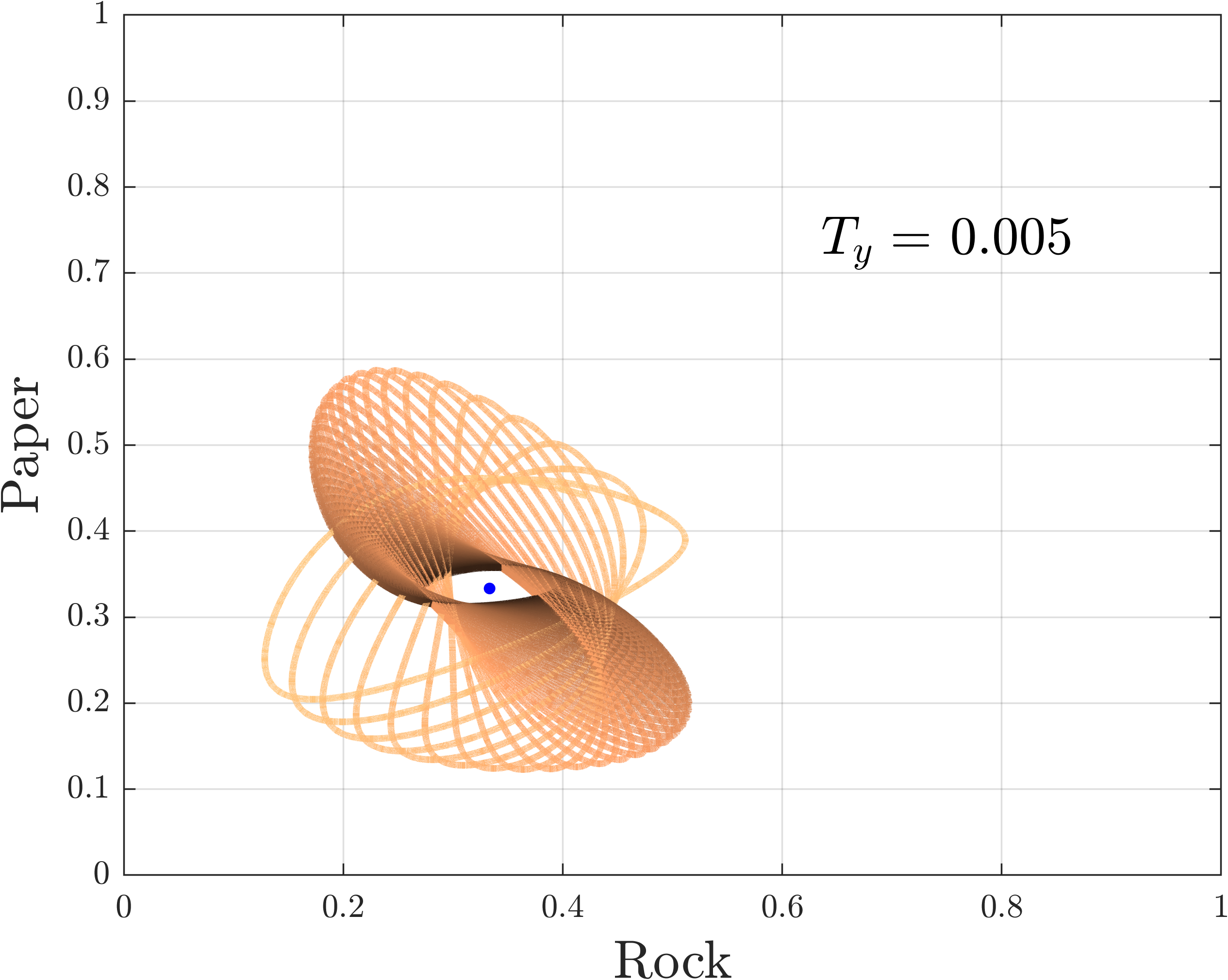}\hspace{5pt}
    \includegraphics[width=0.32\linewidth]{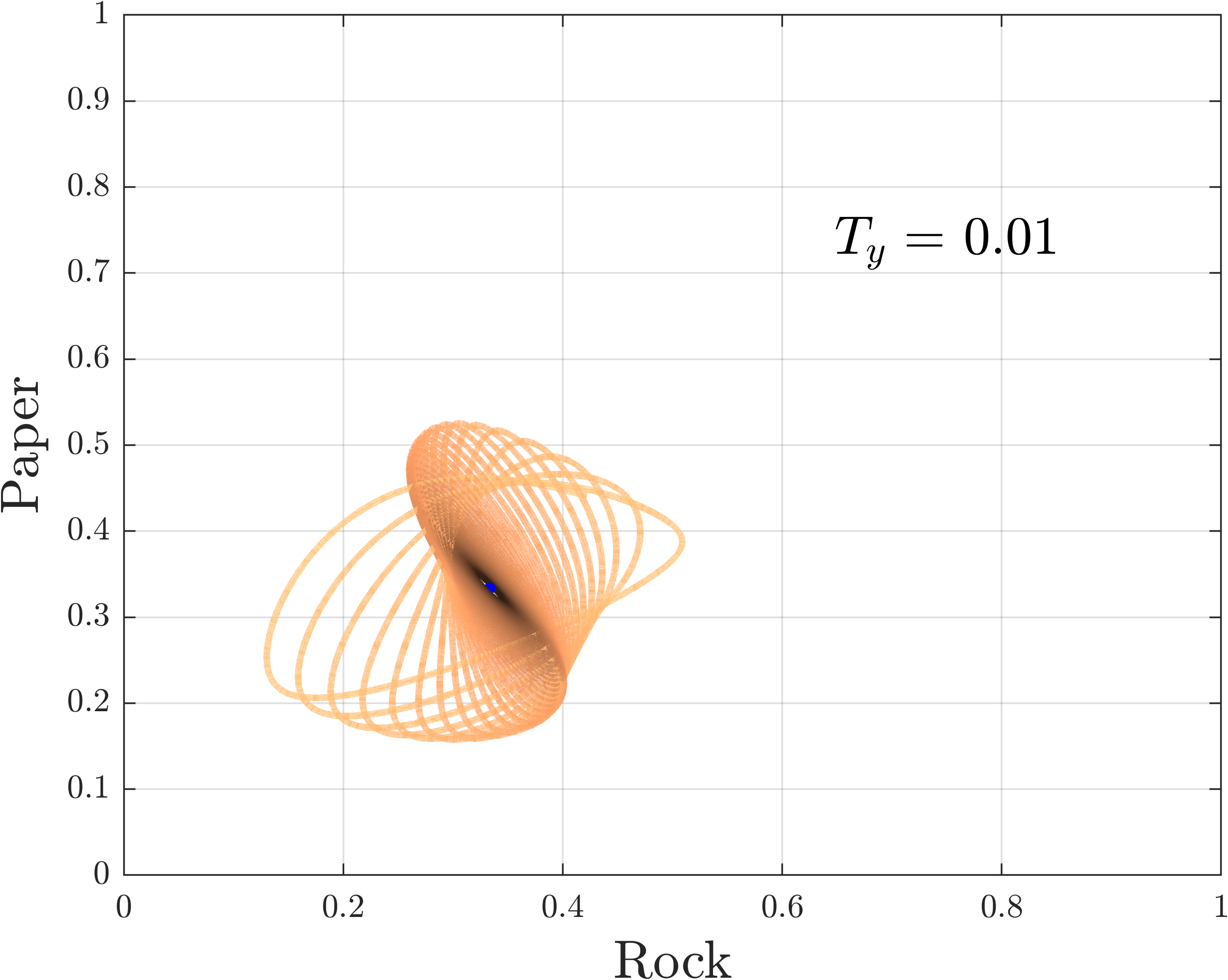}\\[0.2cm]
    \includegraphics[width=0.32\linewidth]{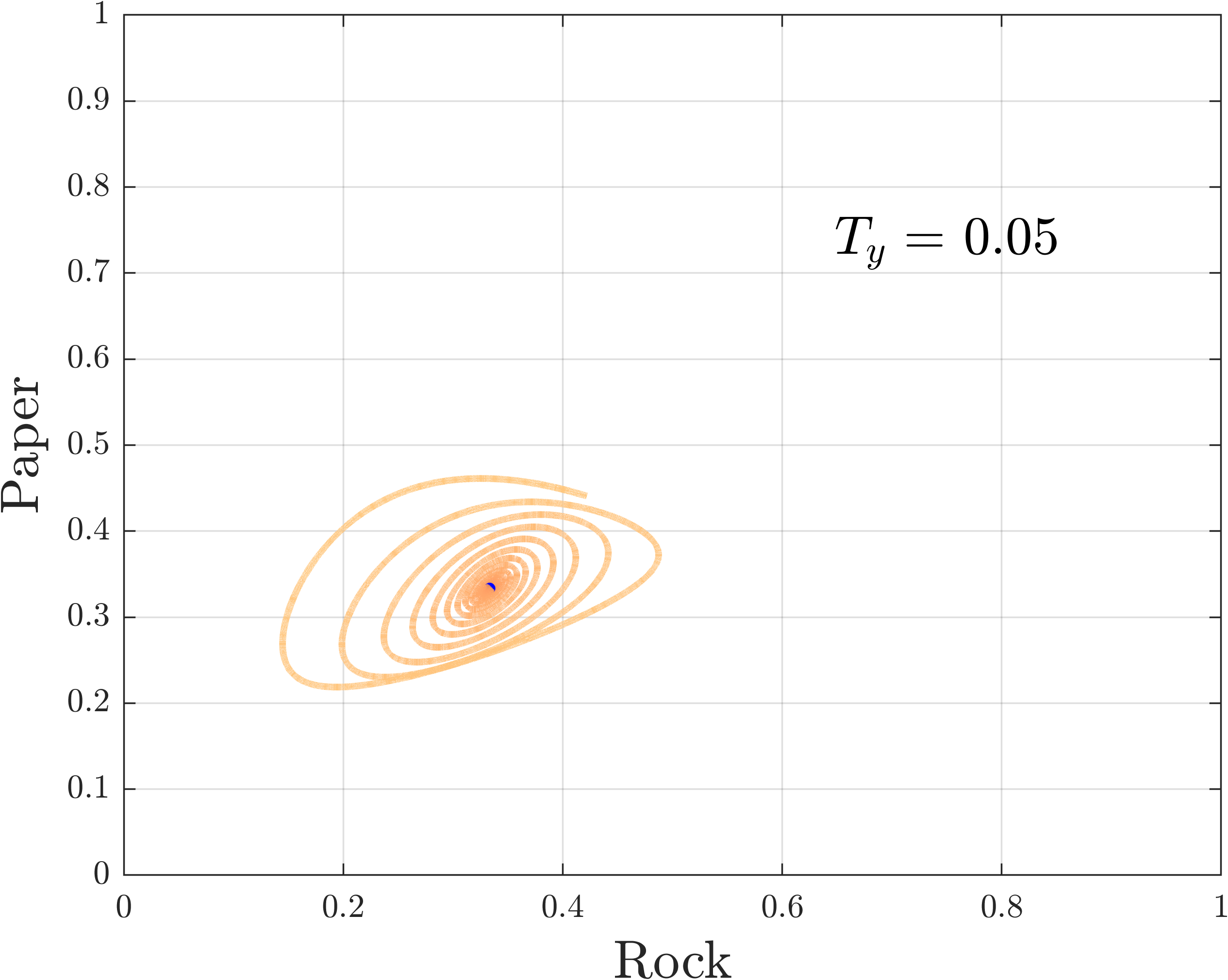}\hspace{5pt}
    \includegraphics[width=0.32\linewidth]{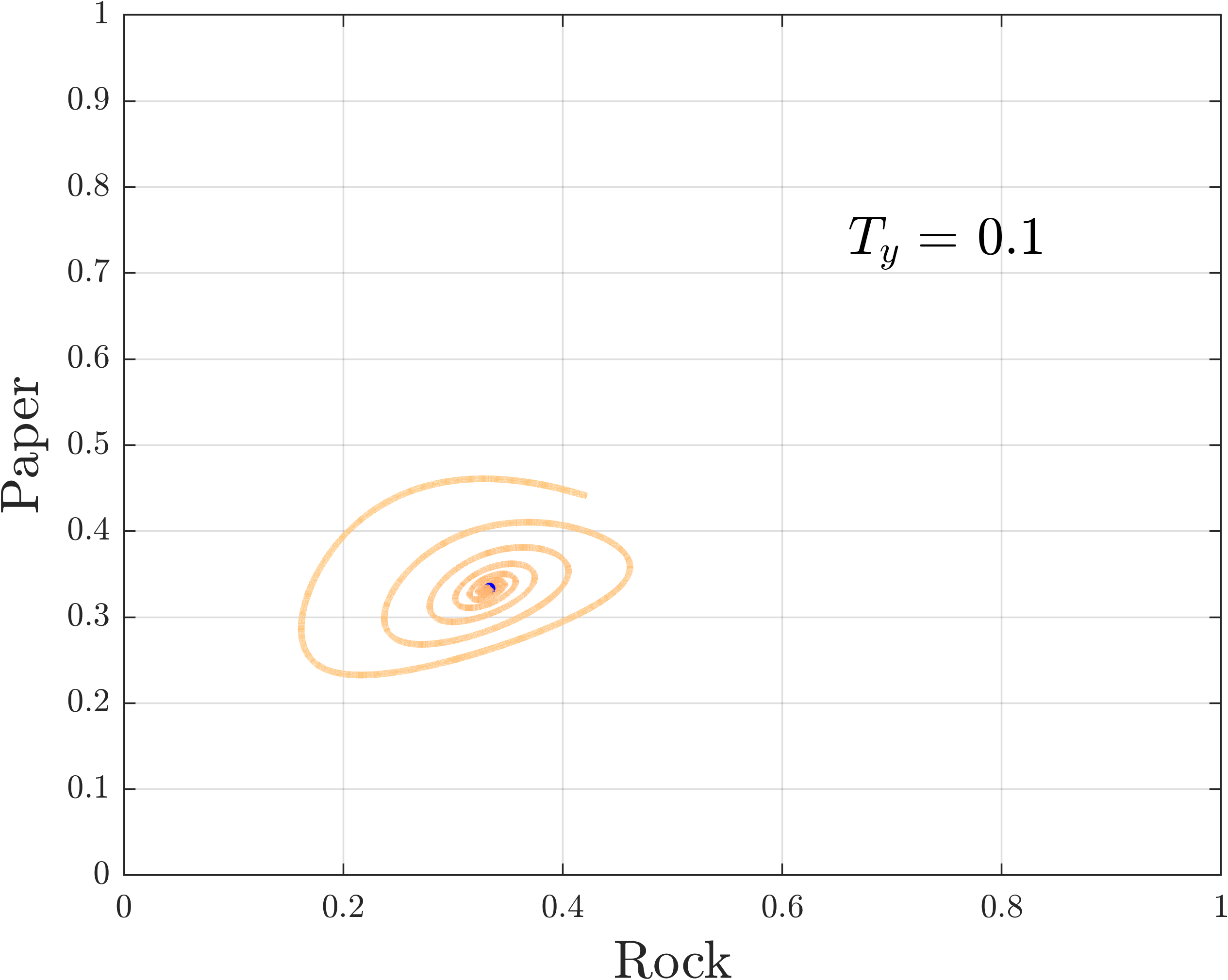}\hspace{5pt}
    \includegraphics[width=0.32\linewidth]{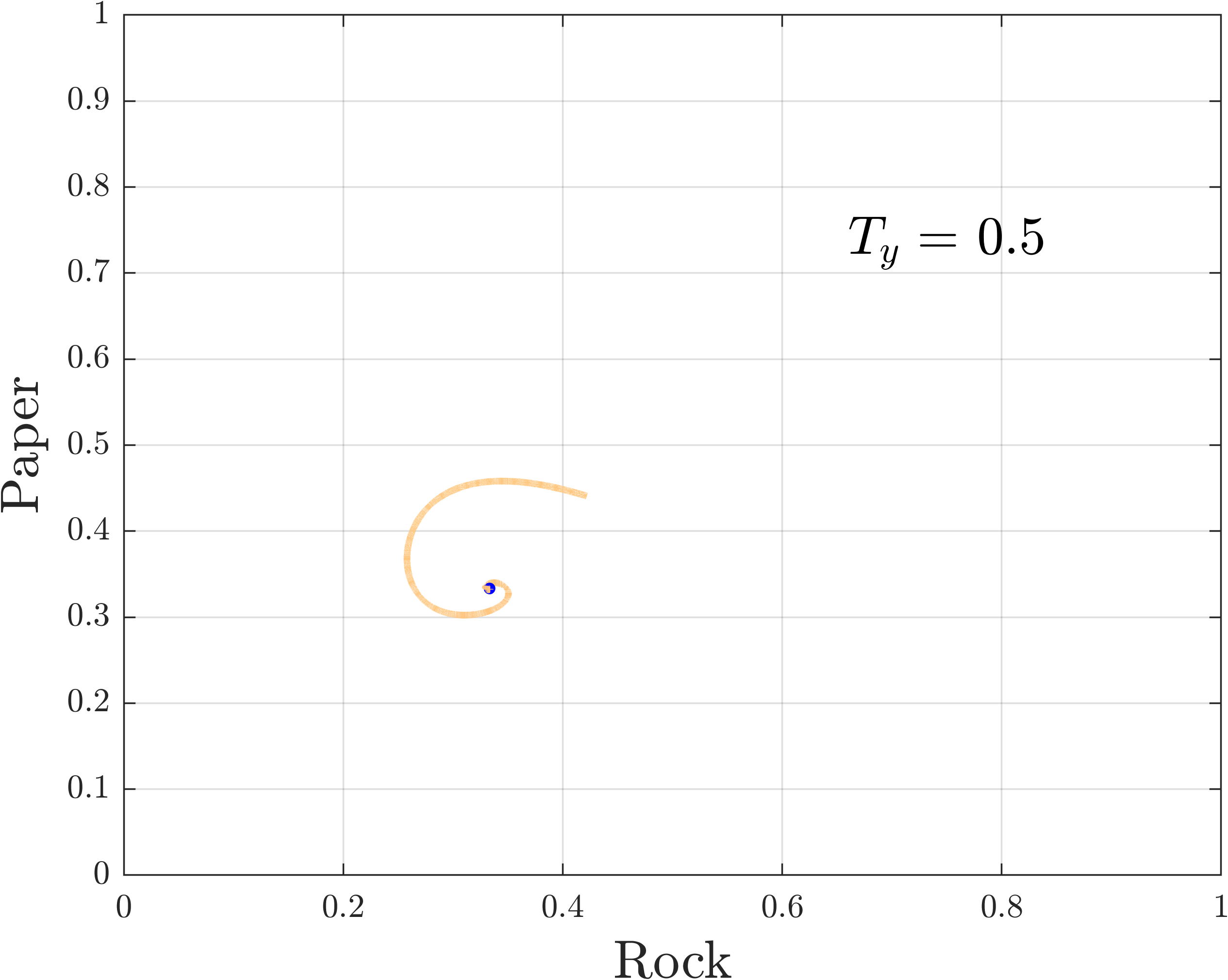}
    \caption{Q-learning dynamics in Rock-Paper-Scissors for $T_x=0$ (no exploration by $x$-agent) and six different exploration rates, $T_y\ge0$ by $y$-agent. The light to dark trajectories (with darkening color indicating increasing time) show the choice distribution for the $x$-agent in the Rock-Paper space.}
    \label{fig:rps}
\end{figure}

The trajectories in \Cref{fig:rps} have been generated for random initial conditions (similar plots are obtained for any other initial condition) and show the choice distribution of the $x$-agent at each time point of the simulation (we used $2\times10^7$ iterations with a step $0.0003$ for the discretization of the continuous time ODE in equation \eqref{eq:kdynamics}). We obtain similar plots for the exploring agent ($y$-agent), see \Cref{fig:rps_y}. When $T_y=0$, the Q-learning dynamics reduce to the replicator dynamics and we recover their cyclic behavior (Poincar{\'e} Recurrence) around the unique interior Nash equilibrium (upper left panel) (see \cite{Per20} and references therein). In all other cases, exploration by one agent suffices for the convergence of the joint-learning dynamics as in the AMPs game (in panels 2 and 3, the dynamics spiral inwards and will eventually converge to the QRE (blue dot)). As we saw in \Cref{exp:network} this is in sharp contrast to the \eqref{eq:network} game in which exploration (even) by several agents was not sufficient for the convergence of Q-learning to a single QRE. It is worth mentioning that this behavior of the Q-learning dynamics in RPS does not rely on the symmetry of the game. We obtain similar plots for the a modified RPS game (with asymmetric Nash equilibrium) (not depicted here).

\begin{figure}[!tb]
    \centering
    \includegraphics[width=0.32\linewidth]{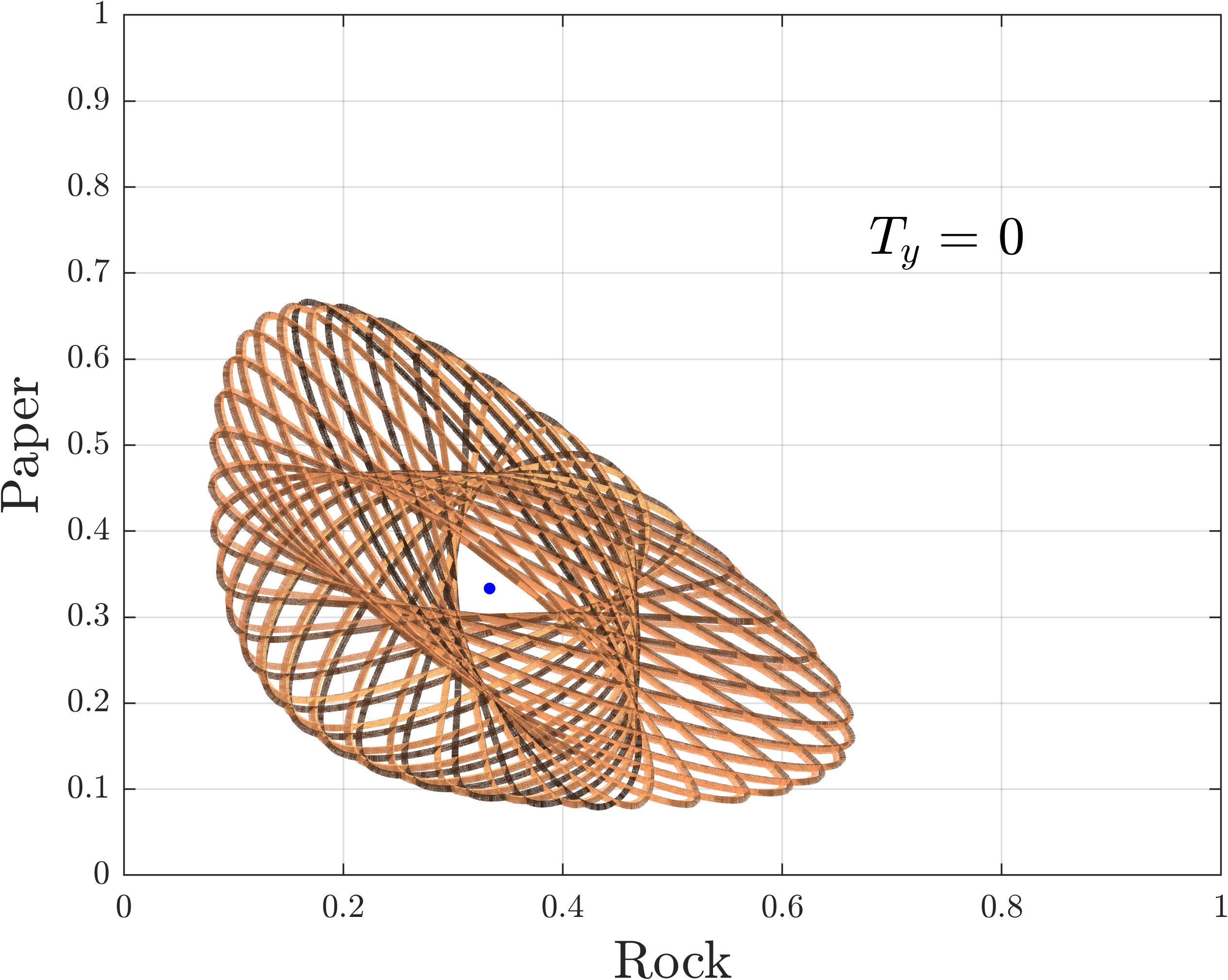}\hspace{5pt}
    \includegraphics[width=0.32\linewidth]{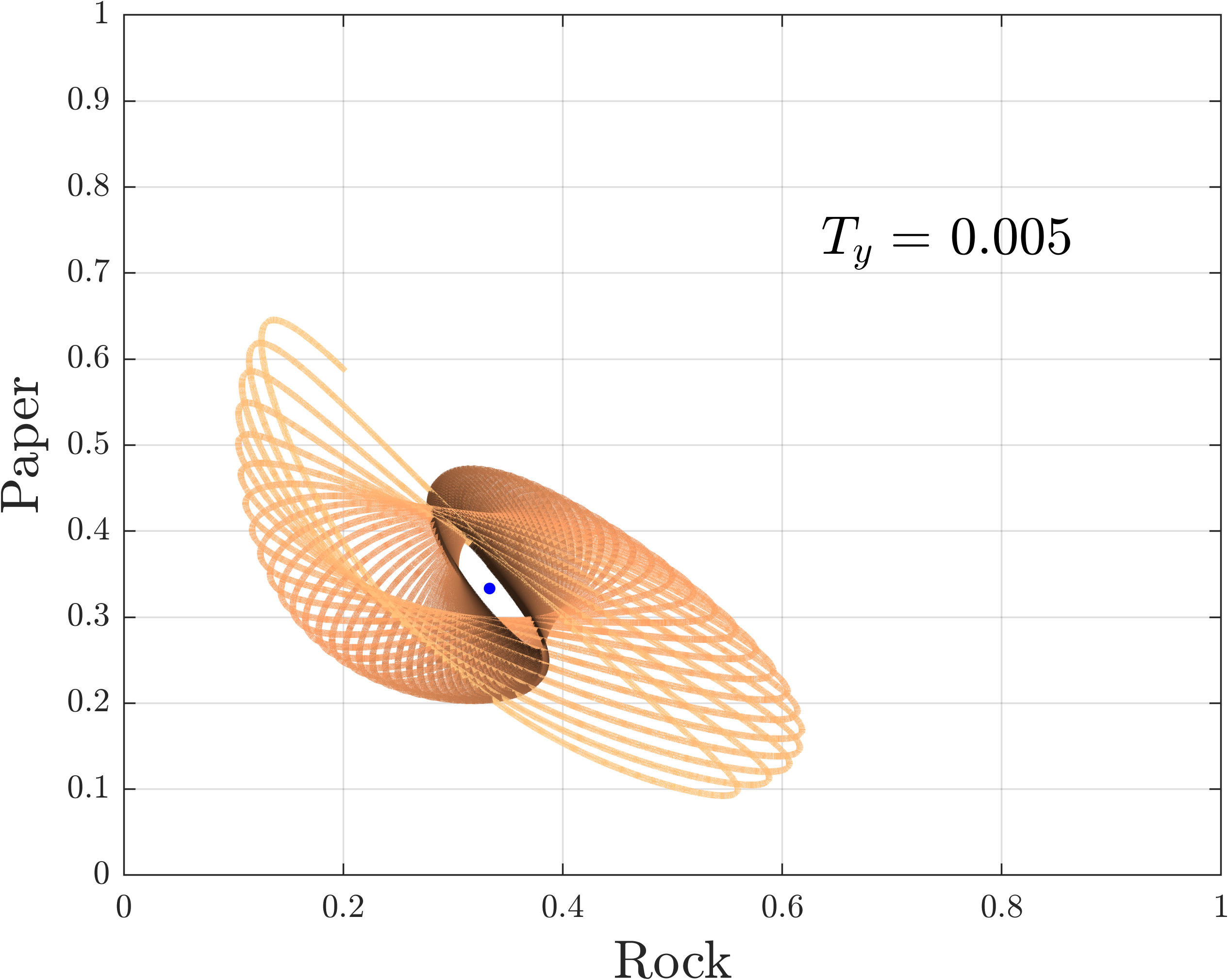}\hspace{5pt}
    \includegraphics[width=0.32\linewidth]{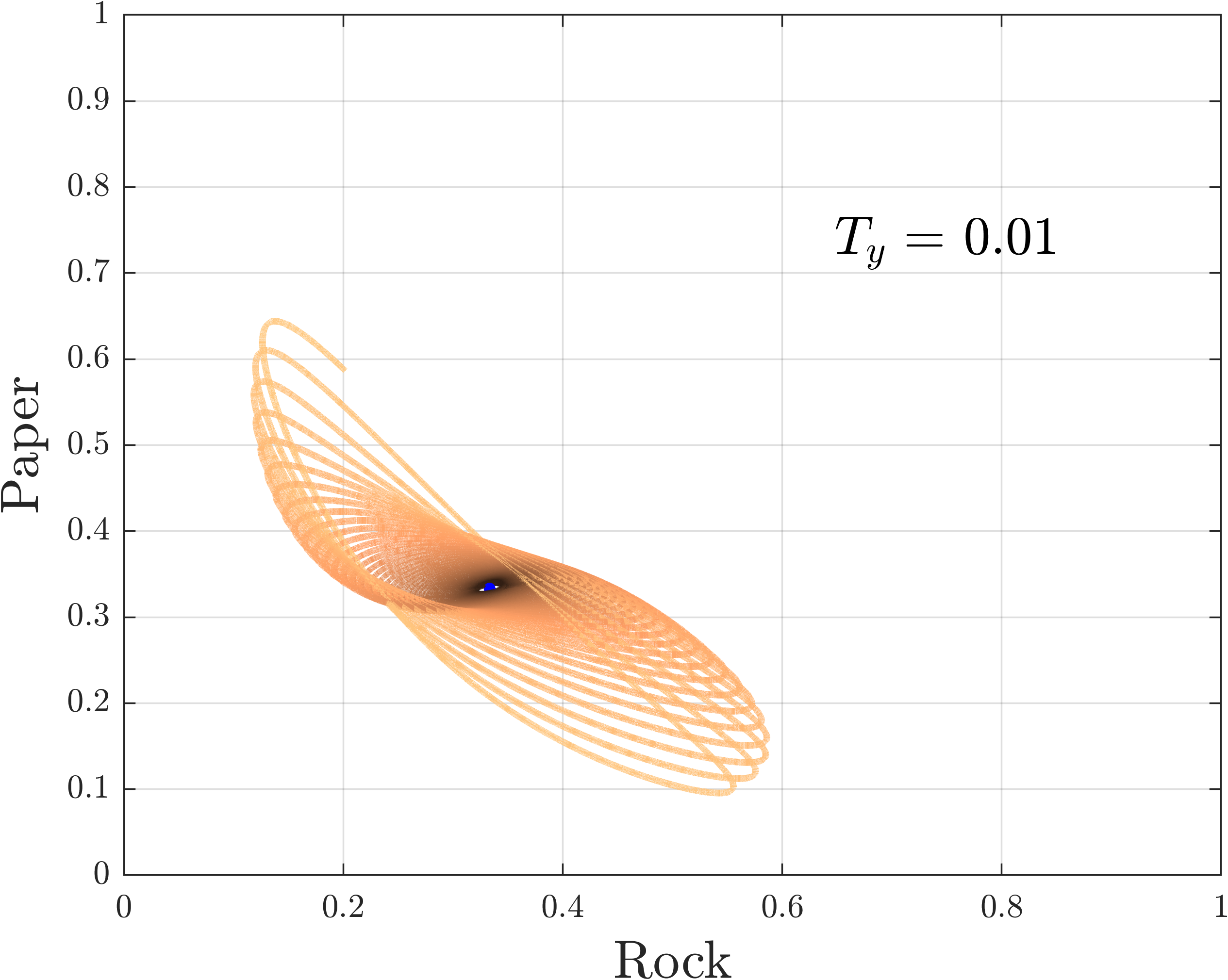}\\[0.2cm]
    \includegraphics[width=0.32\linewidth]{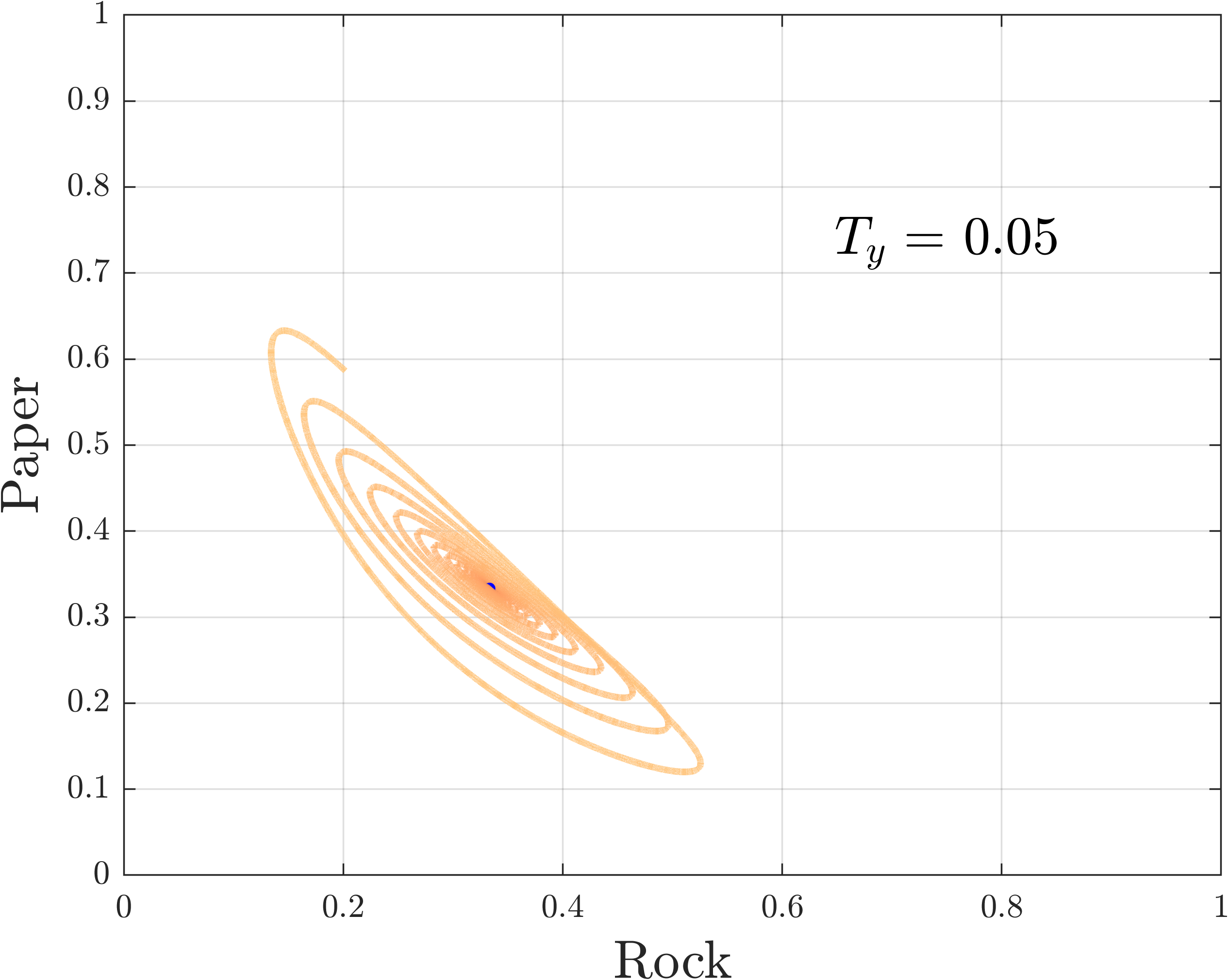}\hspace{5pt}
    \includegraphics[width=0.32\linewidth]{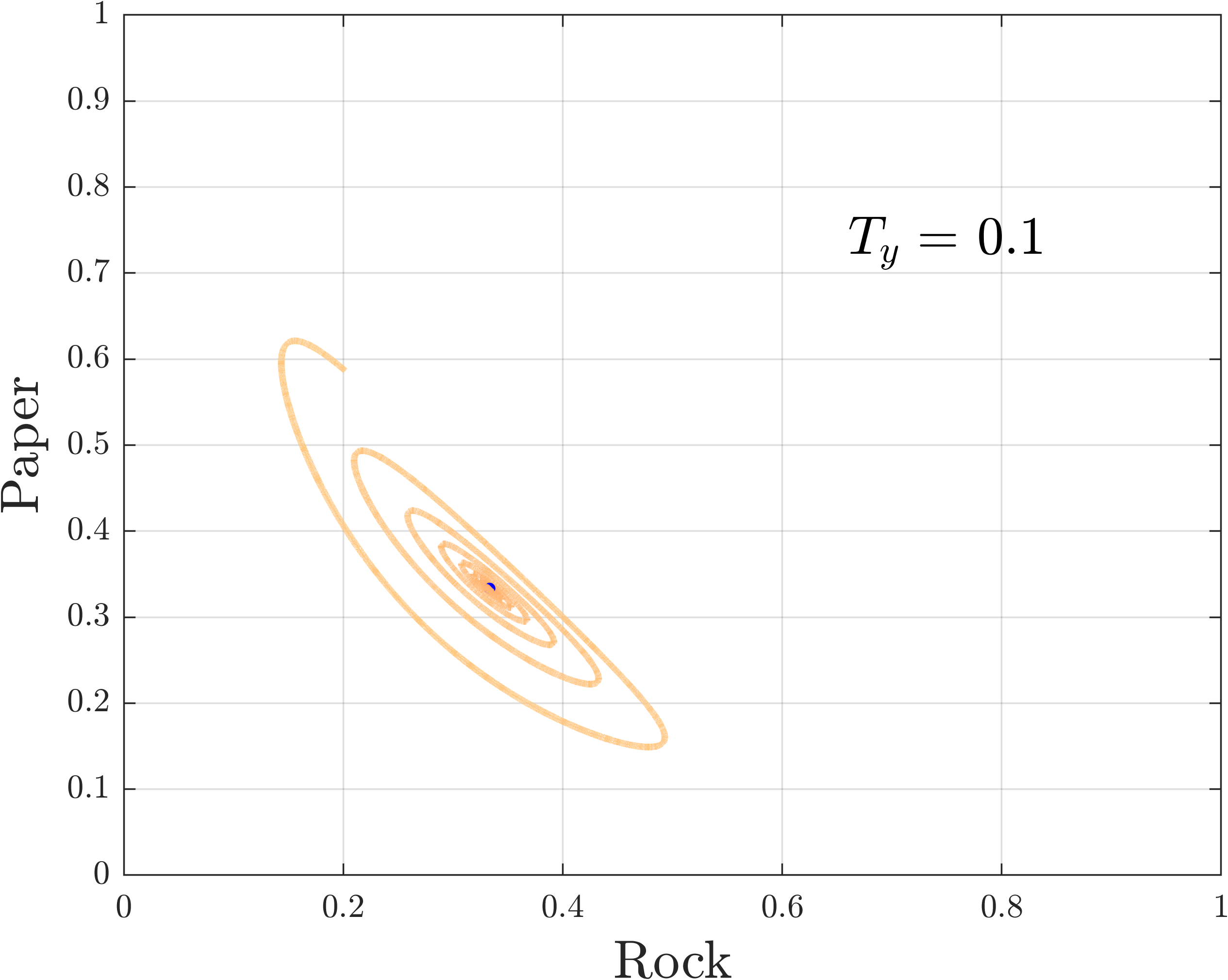}\hspace{5pt}
    \includegraphics[width=0.32\linewidth]{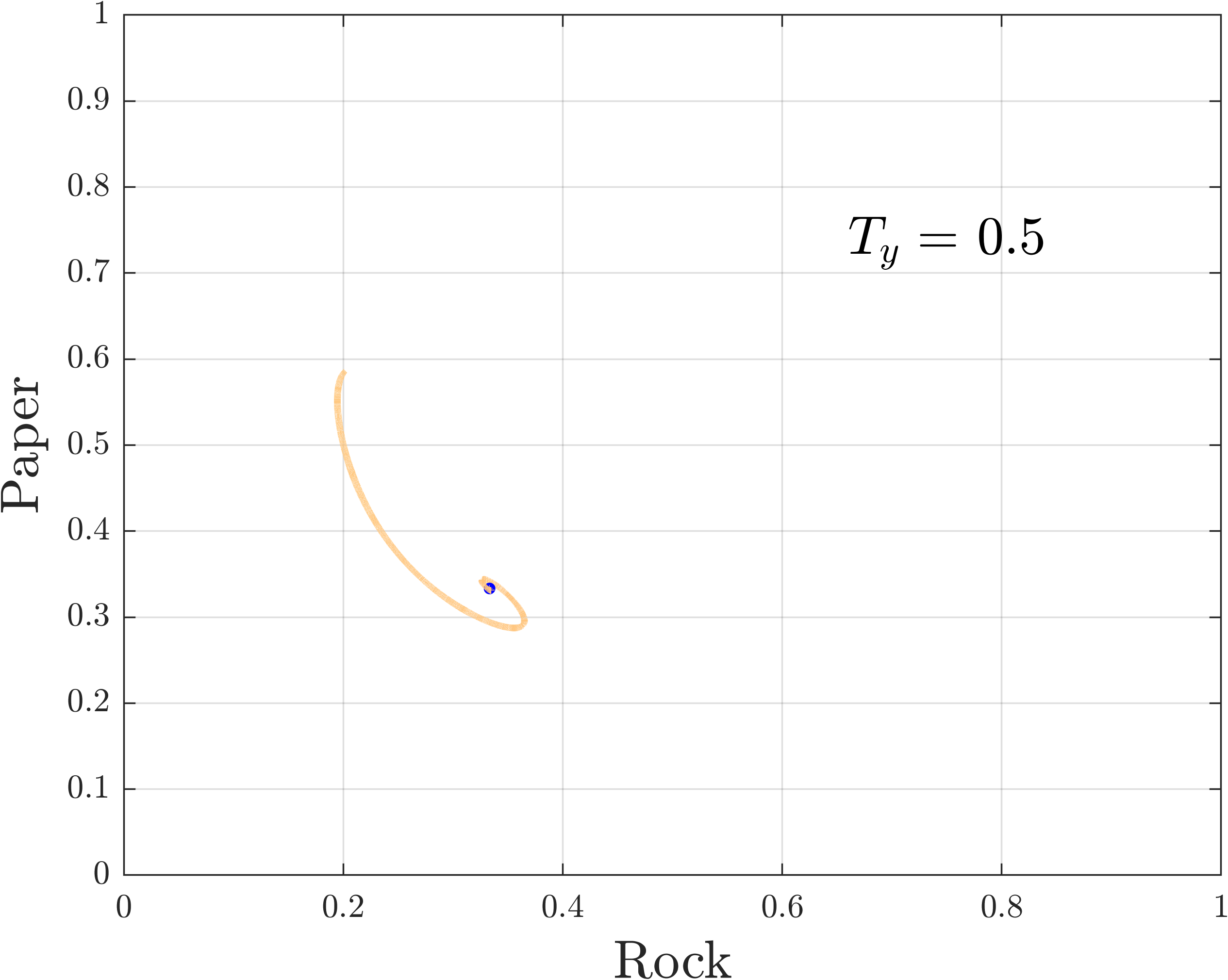}
    \caption{Q-learning trajectories for the $y$-agent in the instances of the RPS game that are shown in \Cref{fig:rps}.}
    \label{fig:rps_y}
\end{figure}
 
\subsection{Edge Case: Exploration by one player in \texorpdfstring{$2x2$}{lll} games}
Up to now, we have treated the case of exploration by only agent experimentally (cf. \Cref{fig:zero,fig:rps,fig:rps_y}). In this part, we consider the edge-case in which only one of two players is exploring in $2x2$ games which is also analytically tractable. Our result is presented in \Cref{prop:2by2} and the ensuing intuition is summarized in \Cref{rem:intuition}. For the following (technical) calculations, we will use the notation 
\[\A=\begin{pmatrix}a_{11}& a_{12}\\ a_{21} & a_{22}\end{pmatrix}, \quad \B=\begin{pmatrix} b_{11} & b_{12} \\ b_{21} & b_{22}\end{pmatrix}.\]

We assume that the game has a (unique) interior Nash equilibrium. This implies (without loss of generality, see e.g., \cite{Pan17}) that 
\begin{equation}\label{eq:interior}
a_{11}>a_{21}, \;a_{12}<a_{22}, \;\;\text{ and }\;\; b_{11}<b_{21},\; b_{12}>b_{22}.
\end{equation}
By letting $a_1:=a_{12}-a_{22}$ and $a_2:=a_{12}+a_{21}-a_{11}-a_{22}$ with $a_1,a_2<0$ and $b_1:=b_{12}-b_{22}$ and $b_2:=b_{12}+b_{21}-b_{11}-b_{22}$ with $b_1,b_2>0$, we have that the unique Nash equilibrium, $(\p,\q)$, with $\p=(p,1-p),\q=(q,1-q)$ of this game is given by $q=\frac{a_1}{a_2}$ and $p=\frac{b_1}{b_2}$ with $p,q\in (0,1)$ by assumption.\par
Since players have two strategies, we can represent their mixed strategies by the vectors $\x=(x,1-x)$ and $\y=(y,1-y)$ with $x,y \in[0,1]$. Using this notation and under the assumption that $T_x=0$, i.e., that the $x$ player is not exploring, the dynamics for the $x$ player in equation \eqref{eq:kdynamics} take the form
\begin{align*}
\dot x &= x \lt (1,\; 0)\A \dbinom{y}{1-y}-(x,\; 1-x)\A\dbinom{y}{1-y}\rt = x(1-x)\lt a_1-a_2y\rt
\intertext{ by assumption \eqref{eq:interior}. Similarly, the dynamics for the $y$ player in equation \eqref{eq:kdynamics} take the form}
\dot y &= y \lt (1,\; 0)\B \dbinom{x}{1-x}-(y,\; 1-y)\B\dbinom{x}{1-x}+T_y(y\ln{y}+(1-y)\ln{1-y}-\ln{y})\rt \\&= y(1-y)\lt b_1-b_2x+T_y\ln{\<\frac1y-1\>}\rt
\end{align*}
where we defined $b_1:=b_{12}-b_{22}$ and $b_2:=b_{12}+b_{21}-b_{11}-b_{22}$ with $b_1,b_2>0$ by assumption \eqref{eq:interior}. For the $y$ player, we will consider various values for $T_y\ge0$. Putting these together, we obtain the system
\begin{align}\label{eq:2by2}
\dot x &= x(1-x)\lt a_1-a_2y\rt\nonumber\\
\dot y &= y(1-y)\lt b_1-b_2x+T_y\ln{\<\frac1y-1\>}\rt
\end{align}
with $x,y\in[0,1]$ and $T_y\ge0$. It is immediate that all four points $(x,y)\in \{0,1\}\times \{0,1\}$ are resting point of the system. All these points lie on the boundary. For the system to have an interior resting point, we have the conditions
\begin{equation}\label{eq:conditions}
y=\frac{a_1}{a_2}, \quad \text{ and } \quad x =\frac{1}{b_2}\lt b_1+T_y\ln{\<\frac{a_2}{a_1}-1\>}\rt. 
\end{equation}
The condition for $y$ is always satisfied by assumption \eqref{eq:interior}. In particular, this yields the Nash equilibrium strategy for the $y$ player. However, the condition for $x$ may yield an $x$ that does not lie within $(0,1)$. For $T_y=0$, the condition becomes $x=b_1/b_2$ which is the Nash equilibrium strategy for the $x$ player. By assumption, this is strictly between $0$ and $1$. However, $x$ depends linearly on $T_y$ and depending on whether $a_2/a_1>2$ or $a_2/a_1<2$ it either increases or decreases in $T_y$.\footnote{The case $a_2=2a_1$ is trivial since it implies that $y=1/2$ and $x=b_1/b_2$ regardless of $T_y$.} Thus, there exists a critical threshold, $T_y^*$, for which $x$ hits the boundary of $[0,1]$, i.e., it either becomes $0$ or $1$. Assume that $x=1$ without loss of generality. At that point, the upper equation in \eqref{eq:2by2} is satisfied regardless of whether $y=a_1/a_2$. Thus, we turn to the lower equation in \eqref{eq:2by2} to obtain a condition for $y$, which yields  
\[b_1-b_2\cdot 1+T_y\ln{\<\frac1y-1\>}=0 \iff y=\lt 1+\exp{\<(b_2-b_1)/T_y\>}\rt^{-1}.\]
Note, that in all the above cases, if an interior resting point exists (for either or both players), then it is unique. We summarize our findings for the $2\times 2$ case in \Cref{prop:2by2}. 

\begin{proposition}\label{prop:2by2}
Let $\Gamma=(\{1,2\},\A,\B)$ with $\A,\B\in \mathbb R^{2\times2}$ be a two-player, two-strategy ($2\times2$) game with $a_2>2a_1$ and a unique interior Nash equilibrium, i.e., such that condition \eqref{eq:interior} holds, and let \[T_y^{\text{crit}}:=(b_2-b_1)\cdot \lt\ln{\<\frac{a_2}{a_1}-1\>}\rt^{-1}.\]
If $T_x=0$, then for any interior starting point $(x_0,y_0)\in(0,1)$, the fixed points $(\p,\q)=((p,1-p), (q,1-q))$ with $p,q\in[0,1]$ of the Q-learning dynamics 
\begin{align*}
\dot x &= x(1-x)\lt a_1-a_2y\rt\nonumber\\
\dot y &= y(1-y)\lt b_1-b_2x+T_y\ln{\<\frac1y-1\>}\rt
\end{align*}
depend on $T_y$ as follows:
\begin{itemize}[noitemsep,leftmargin=*]
\item if $T_y=0$, then the dynamics are cyclic, i.e., they do not have a resting point.
\item if $0<T_y< T_y^{\text{crit}}$, then they have a unique interior resting point which is given by
\[(p,q)=\<\frac{1}{b_2}\lt b_1+T_y\ln{\<\frac{a_2}{a_1}-1\>}\rt,\;\frac{a_1}{a_2}\>\;\; \text{ with } p,q \in (0,1).\]
\item otherwise, i.e., if $T_y>T_y^{\text{crit}}$, then they have a resting point that lies in the interior only for the exploring player which is given by
\[(p,q)=\<1, \; \lt 1+\exp{\<(b_2-b_1)/T_y\>}\rt^{-1}\>, \;\; \text{ with } q\in (0,1).\]
In particular, as $T_y\to\infty$, it holds that $q\to1/2$.
\end{itemize}
\end{proposition}
If $a_2<2a_1$, then $x$ is decreasing in $T_y$ and hence, at $T_y^{crit}$ it becomes $0$ (rather than $1$). In this case, solving equation \eqref{eq:conditions} for $T_y^{\text{crit}}$ yields $T_y^{\text{crit}}=-b_1\cdot \lt\ln{\<\frac{a_2}{a_1}-1\>}\rt^{-1}$. By renaming the strategies of the $x$ player, this case is equivalent to the one presented in \Cref{prop:2by2}. The general case of more than $2$ strategies for each player is qualitatively equivalent to the case presented here.
\begin{remark}[Intuition of \Cref{prop:2by2}]\label{rem:intuition}
The main takeaway of \Cref{prop:2by2} is the qualitative description of the resting points of the Q-learning dynamics for any interior starting point. When $T_y=0$, the dynamics are the replicator dynamics, which are well-known to cycle around the unique interior Nash equilibrium in this case \cite{Bai18}. If $T_y>0$, then this suffices to ensure convergence to the unique interior QRE (cf. \Cref{thm:main}). As long as $T_y$ is small enough, i.e., smaller than a critical value, $T_y^{\text{crit}}$, the $y$ component of the QRE corresponds precisely to the Nash equilibrium (mixed) strategy for the $y$-player (the exploring player) and an interior value for the $x$-player (different that her Nash equilibrium mixture). This value is increasing (assuming that $a_2>2a_1$, otherwise it is decreasing or constant if $a_2=2a_1$) in the exploration rate $T_y$ of the $y$-player. This implies, that the QRE component for the $x$-player approaches the boundary. After exploration by the $y$-player exceeds the critical threshold $T_y^{\text{crit}}$, the $x$-player starts playing a pure strategy. At that point onward (i.e., for larger exploration rates), the utility of the $y$ player is dominated by her exploration term and her mixture (at QRE) starts to approach the uniform distribution.  
\end{remark}
The result of \Cref{prop:2by2} is illustrated in the main part of the paper via Asymmetric Matching Pennies (AMPs) game in \Cref{fig:zero}.

\subsection{3D Visualization of the Lyapunov function in \texorpdfstring{$n$}{kkk}-agent network games}\label{sub:visualization}
To visualize the Lyapunov function in zero-sum network games with $n$-agents with strictly positive exploration profiles, i.e., $T_k>0$ for all $k\in V$, (KL-divergence from the current action profile to the unique QRE for that exploration profile), we adapt the dimension reduction method of \cite{Li18} (cf. \cite{Leo21}). This yields panel 4 in \Cref{fig:line_network}.\par
For an $n$-agent network game, with a fixed and strictly positive exploration profile, we start by determining its unique QRE, $q$. By our main result, \Cref{thm:main}, this can be done by simulating the Q-learning dynamics. Then, we select two vectors $u,v$ with $n$ random entries each in $(0,1)$. The random entries of the vectors $u,v$ correspond to the probability with which each agent selects their first action (here $H$). Instead of forming convex combinations of these random vectors and plotting the Lyapunov function across this (randomly selected) space, we perform the following transformation that allows for a more comprehensive snapshot of the whole joint action space. Specifically, we map each coordinate $u(k)$ (and similarly for $v$) with $k\in V$ to $\tilde{u}(k) = \ln{u(k)/(1-u(k))}$. Then, we form linear combinations of the transformed vectors $\tilde{u},\tilde{v}$ using real-valued scalars,$\alpha, \beta\in \mathbb R$. Note that the all-zero vector (in the transformed space) corresponds to the uniform distribution for each agent (in the choice distribution space). Finally, we map the resulting point $z:=\alpha\cdot \tilde{u}+\beta\cdot \tilde{v}$ from the transformed space back to the product simplex via the (coordinate-wise) transformation $z(k)\to \exp{(z(k))}/(\exp{(z(k))}+1)$ and plot the potential at the resulting point (KL-divergence between that point and the unique QRE, $q$). We repeat the process for a range of both positive and negative values for $\alpha$'s and $\beta$'s. This yields the $x-y$ coordinates in panel 4 and the evaluation yields the depicted 3D surface. The process is summarized in \Cref{alg:visual}.

\begin{algorithm}[!bth]
\caption{3D Visualization of the Lyapunov function (KL-divergence)}\label{alg:visual}
\vspace*{0.1cm}
\raggedright
\textbf{Input (network game):} number of agents, payoff matrices, (strictly positive) exploration rates.\\
\textbf{Output:} Snapshot of the Lyapunov function (KL-divergence).\\[-0.3cm]

\begin{algorithmic}[1]
\Procedure{Compute QRE}{$T_k,k=1,\dots,n$}
\State{$q\gets$ unique QRE (e.g., by running \eqref{eq:kdynamics})}
\EndProcedure 
\State{$u,v$ generate random vectors with $n$ entries in $(0,1)$.}

\Procedure{Transform Variables}{$u,v, \alpha,\beta$}
\For {$k\gets V$}
\State {$u(k) \gets \ln{\<u(k)/(1-u(k))\>}$}
\State {$v(k) \gets \ln{\<v(k)/(1-v(k))\>}$}
\EndFor
\State {$z \gets \alpha\cdot u+\beta\cdot v$}
\For {$k\gets V$}
\State {$z(k) \gets \exp{\<z(k)/(z(k)+1)\>}$}
\EndFor
\EndProcedure
\Procedure{Evaluate Lyapunov function}{$q,z$}
\State{$\KL=0$}
\For{$k\in V$}
\State{$\KL\gets \KL+
q(k)\ln\frac{q(k)}{z(k)}+(1-q(k))\ln\frac{1-q(k)}{1-z(k)}$}
\EndFor
\Return tuple ($\alpha,\beta,\KL$) 
\EndProcedure
\State{\textbf{plot} $\gets$ ($\alpha,\beta,\KL$)}
\end{algorithmic}
\end{algorithm}

A restriction of this method in $n$-player games is that, for each point that it generates, it uses the same $\alpha,\beta$ to scale the transformed variables of \emph{all} agents. 
\Cref{fig:lyapunov_sup} shows snapshots of the Lyapunov function in four instances of the \eqref{eq:network} game with fixed (and strictly positive) exploration profiles for different numbers of agents.

\begin{figure}[!tb]
    \centering
    \includegraphics[width=0.97\linewidth,clip=true, trim=5.2cm 0 4.3cm 0]{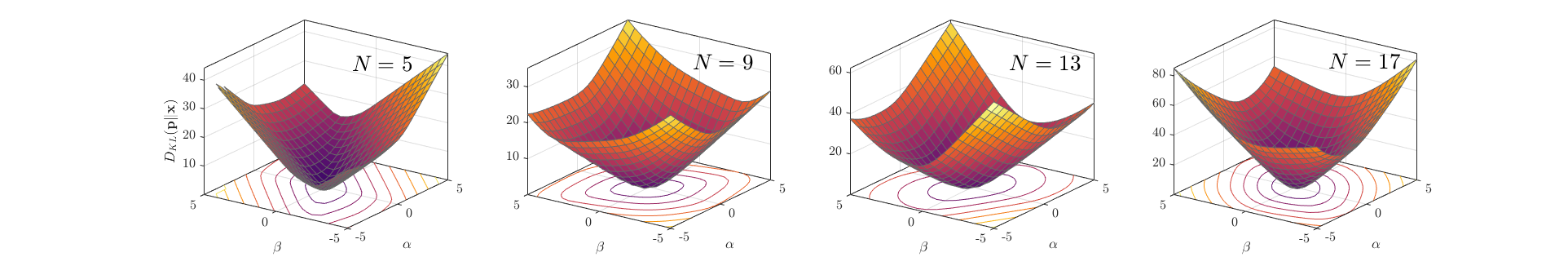}
    \caption{Snapshots of the Lyapunov function (KL-divergence between each choice distribution profile and the unique QRE) in four instances of the \eqref{eq:network} game with fixed (and strictly positive) exploration profiles and different numbers of agents. In all cases, the Lyapunov (potential) function is convex with a unique minimizer at $0$.}
    \label{fig:lyapunov_sup}
\end{figure}

\subsection{Equilibrium selection in the network game}

One question that is hard to tackle theoretically concerns the equilibrium selection as exploration rates converge to $0$. As we saw in \Cref{thm:main}, when all agents have positive exploration rates, then there is a unique QRE and the joint-learning dynamics converge to that QRE. However, as exploration rates approach zero (for instance, after the exploration phase ends for all agents), it is not clear which equilibrium will be selected in the original game (as the limit of the sequence of the unique QRE for the different strictly positive exploration profiles). \par
In this part, we test this question experimentally in the \eqref{eq:network} game of \Cref{exp:network}. We consider an instance with $3$ non-dummy agents and different exploration policies for the agents. Recall that in this case, the original network game (with no exploration) has multiple Nash equilibria of the following form: the odd agents ($p_1$ and $p_3$) select $T$ with probability 1 and the even agent ($p_2$) plays an arbitrary strategy in $(0,1)$ (i.e., probabilities of playing $H$). In any equilibrium, the payoff of $p_2$ is $0$ whereas the payoffs of $p_1$ and $p_3$ sum up to $2$. However, the crucial point is that the split of $2$ between $p_1$ and $p_3$ critically depends on the strategy of $p_2$. In particular, we saw in \Cref{fig:summary} that without exploration, the Q-learning dynamics can converge to any of these multiple equilibria, thus inducing arbitrary asymmetries between the payoffs of $p_1$ and $p_3$.\par
The results of one representative exploration scenario (with linearly changing exploration rates) and averages over 50 runs with randomly matched CLR-1 and ETE exploration policies are presented in \Cref{fig:selection}.
\begin{figure}[!tb]
    \centering
    \includegraphics[width=0.98\linewidth]{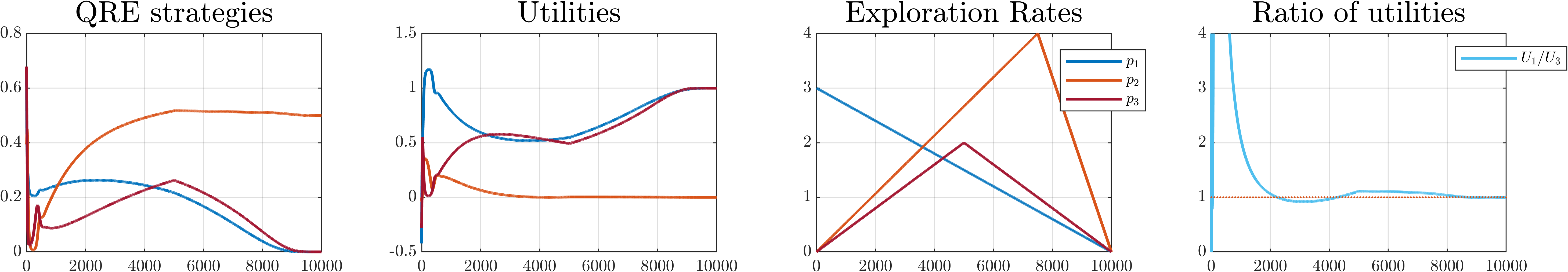}\\[0.2cm]
    \includegraphics[width=0.98\linewidth]{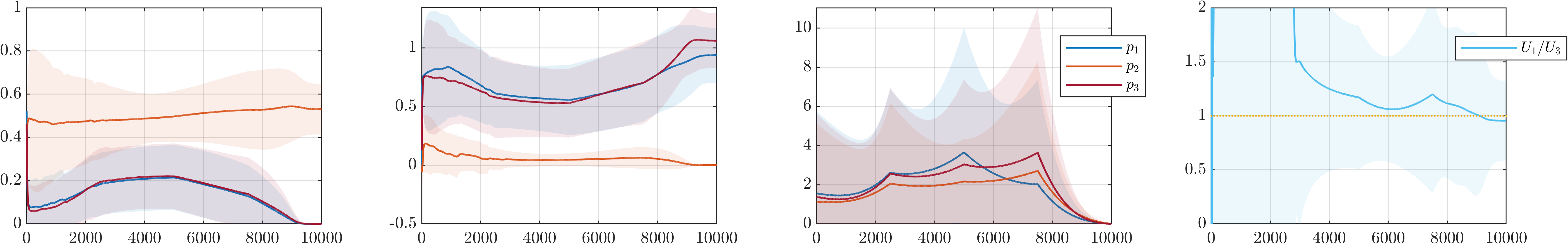}
    \caption{Effects of exploration on equilibrium selection in the \eqref{eq:network} of \Cref{exp:network}. The upper panels show an individual run and the bottom panels show averages (means and 1 standard deviation as shaded region around the mean) over 50 runs. Panels 1 to 3 show the probability of playing $H$ at QRE, the utilities and the exploration rates of the agents, respectively. The effect of exploration is shown via the ratio of utilities of $p_1$ and $p_3$ in the fourth panel of each row. Exploration by the even agent ($p_2$) leads that agent to select the $0.5$ strategy at equilibrium (when exploration drops back to $0$ by all agents) which results in a fair split (close to 1, dotted red line in panels 4) of the payoffs between agents $p_1$ and $p_3$.}
    \label{fig:selection}
\end{figure}

The main takeaway of these experiments is captured by the last panel "Ratio of utilities" of each row. Namely, sufficient exploration by the even agent (the agent with multiple equilibrium strategies) leads that agent to select a strategy close to the uniform one (here $0.5$ since there are two actions). In turn, this leads to a fair split of the stake between the odd numbered agents ($p_1$ and $p_3$). This is in sharp contrast to the case without exploration (cf. \Cref{fig:summary} in the main part) in which any equilibrium (i.e., any strategy between $0,1$) is a potential limit point of the dynamics for the even agent (thus, leading to arbitrary splits of the share between the odd numbered agents). \par
While these results cannot lead to a formal argument about the effect of exploration in equilibrium selection (when exploration goes back to zero for all agents), they highlight the importance of further studying equilibrium selection in competitive environments both experimentally and theoretically. In particular, the (potentially positive) effects of (individual) exploration to social welfare (here, fair distribution of rewards) constitute an concrete and intriguing direction for future research in this area.

\end{document}